\documentclass[11pt]{article}
\pdfoutput=1
\usepackage[margin=1in]{geometry}

\usepackage{amsmath, amssymb, amsthm, amsfonts, graphicx}
\usepackage{xcolor}

\usepackage{hyperref}
\hypersetup{
	bookmarksnumbered=true, 
	unicode=false, 
	pdfstartview={FitH}, 
	pdftitle={}, 
	pdfauthor={}, 
	pdfsubject={}, 
	pdfcreator={}, 
	pdfproducer={}, 
	pdfkeywords={}, 
	pdfnewwindow=true, 
	colorlinks=true, 
	linkcolor=blue, 
	citecolor=blue, 
	filecolor=blue, 
	urlcolor=blue 
}

\usepackage{complexity}
\usepackage{thmtools}
\usepackage[capitalize,nameinlink]{cleveref} 
\usepackage{verbatim}
\usepackage{qcircuit}
\usepackage{tocstyle}
\usepackage{tikz}
\usepackage{verbatim}

\usepackage{enumitem}
\setlist[itemize]{itemsep=-4pt}

\usetikzlibrary{shapes.misc,shapes.geometric,arrows,positioning,calc,backgrounds,graphs,math,quotes}
\newcommand{\ket}[1]{| #1\rangle}        
\newcommand{\bra}[1]{\langle #1|}        
\newcommand{\braket}[2]{\langle #1 | #2 \rangle} 
\newcommand{\ketbra}[2]{| #1 \rangle\!\langle #2 |} 

\newcommand{\ii}{\mathbb{I}}		

\newtheorem{theorem}{Theorem}

\newtheorem{lemma}[theorem]{Lemma}

\newcommand{\eq}[1]{Eq.~\hyperref[eq:#1]{(\ref*{eq:#1})}}
\renewcommand{\sec}[1]{\hyperref[sec:#1]{Section~\ref*{sec:#1}}}
\newcommand{\app}[1]{\hyperref[app:#1]{Appendix~\ref*{app:#1}}}
\newcommand{\tab}[1]{\hyperref[tab:#1]{Table~\ref*{tab:#1}}}
\newcommand{\fig}[1]{\hyperref[fig:#1]{Figure~\ref*{fig:#1}}}
\newcommand{\figa}[2]{\hyperref[fig:#1]{Figure~\ref*{fig:#1}#2}}
\newcommand{\figx}[2]{\hyperref[fig:#1]{Figure~\ref*{fig:#1}(#2)}}
\newcommand{\thm}[1]{\hyperref[thm:#1]{Theorem~\ref*{thm:#1}}}
\newcommand{\lem}[1]{\hyperref[lem:#1]{Lemma~\ref*{lem:#1}}}
\newcommand{\cor}[1]{\hyperref[cor:#1]{Corollary~\ref*{cor:#1}}}
\newcommand{\defn}[1]{\hyperref[def:#1]{Definition~\ref*{def:#1}}}
\newcommand{\alg}[1]{\hyperref[alg:#1]{Algorithm~\ref*{alg:#1}}}
\newcommand{\prob}[1]{\hyperref[prob:#1]{Problem~\ref*{prob:#1}}}
\newcommand{\threatM}[1]{\hyperref[threat:#1]{Threat Model \ref*{threat:#1}}}

\newcommand{\op}[1]{\operatorname{#1}}
\newcommand{\plog}{\operatorname{polylog}}
\newcommand{\lr}[1]{\left(#1\right)}

\usepackage{stmaryrd}
\usepackage{trimclip}

\makeatletter
\DeclareRobustCommand{\shortto}{%
	\mathrel{\mathpalette\short@to\relax}%
}

\newcommand{\short@to}[2]{%
	\mkern2mu
	\clipbox{{.5\width} 0 0 0}{$\m@th#1\vphantom{+}{\shortrightarrow}$}%
}
\makeatother

\begin{document}

\title{\bfseries Hamiltonian simulation with nearly\\optimal dependence on spectral norm}

\author{Guang Hao Low\thanks{Quantum Architectures and Computation, Microsoft Research, Redmond, Washington, USA}
	\\[.25em]
	\texttt{guanghao.low@microsoft.com}
}

\maketitle
\begin{abstract}
	We present a quantum algorithm for approximating the real time evolution $e^{-iHt}$ of an arbitrary $d$-sparse Hamiltonian to error $\epsilon$, given black-box access to the positions and $b$-bit values of its non-zero matrix entries. The complexity of our algorithm is $\mathcal{O}((t\sqrt{d}\|H\|_{1 \shortto 2})^{1+o(1)}/\epsilon^{o(1)})$ queries and a factor $\mathcal{O}(b)$ more gates, which is shown to be optimal up to subpolynomial factors through a matching query lower bound. This provides a polynomial speedup in sparsity for the common case where the spectral norm $\|H\|\ge\|H\|_{1 \shortto 2}$ is known, and generalizes previous approaches which achieve optimal scaling, but with respect to more restrictive parameters. By exploiting knowledge of the spectral norm, our algorithm solves the black-box unitary implementation problem -- $\mathcal{O}(d^{1/2+o(1)})$ queries suffice to approximate any $d$-sparse unitary in the black-box setting, which matches the quantum search lower bound of $\Omega(\sqrt{d})$ queries and improves upon prior art [Berry and Childs, QIP 2010] of $\tilde{\mathcal{O}}(d^{2/3})$ queries. Combined with known techniques, we also solve systems of sparse linear equations with condition number $\kappa$ using $\mathcal{O}((\kappa \sqrt{d})^{1+o(1)}/\epsilon^{o(1)})$ queries, which is a quadratic improvement in sparsity.
\end{abstract}

\tableofcontents
\newpage
\section{Introduction}
Simulating the dynamics of quantum systems is a major potential application of quantum computers, and is the original motivation for quantum computation~\cite{Feynman1982}. Moreover, the existence of efficient classical algorithms for this problem is unlikely as it is \textsc{BQP}-complete~\cite{Osborne2012}. 
The first explicit quantum simulation algorithm by Lloyd~\cite{Lloyd1996universal} considered Hamiltonians described by a sum of non-trivial local interaction terms, which is relevant to many realistic systems. This was later extended by Aharonov and Ta-Shma~\cite{Aharonov2003Adiabatic} to sparse Hamiltonians, which generalizes local Hamiltonians and is a natural model for designing other quantum algorithms~\cite{Childs2003Exponential,Harrow2009}. 

This paper considers the problem of simulating a $d$-sparse Hamiltonian $H\in\mathbb{C}^{N\times N}$ with at most $d$ nonzero entries in any row. That is, given a description of $H$, an evolution time $t>0$, and a precision $\epsilon>0$, the Hamiltonian simulation problem is concerned with approximating the time-evolution operator $e^{-iHt}$ with an error at most $\epsilon$. In the standard setting following~\cite{Berry2012}, a $d$-sparse $H$ is described by black-box unitary oracles that compute the positions and values of its nonzero matrix entries.  These are more precisely defined as follows.
\begin{restatable}[Sparse matrix oracles]{definition}{sparseoracles}
	\label{def:Sparse_Oracle}
	A $d$-sparse matrix $H\in \mathbb{C}^{N\times N}$ has a black-box description if there exists the following two black-box unitary quantum oracles.
	\begin{itemize}
		\item ${O}_{H}$ is queried by row index $i\in[N]$, column index $k\in[N]$, and $z\in\{0,1\}^b$ and returns ${H}_{ik}=\bra{i}{H}\ket{k}$ represented as a $b$-bit number:
		\begin{align}
		{O}_{H}\ket{i}\ket{k}\ket{z}=\ket{i}\ket{k}\ket{z\oplus {H}_{ik}}.
		\end{align}
		\item ${O}_{F}$\footnote{In the dense or non-sparse case, the oracle ${O}_{F}$ is replaced by identity, which implements an $N$-sparse matrix $H$.} is queried by row index $i\in[N]$ and an index $l\in[d]$, and returns, in-place, the column index $f(i,l)$ of the $l^{\text{th}}$ non-zero matrix entry in the $i^{\text{th}}$ row:
\begin{align}
{O}_{F}\ket{i}\ket{l}=\ket{i}\ket{f(i,l)}.
\end{align}
	\end{itemize}
\end{restatable}

The complexity of $d$-sparse Hamiltonian simulation in terms of these black-box queries has been studied extensively. Early approaches~\cite{Berry2007Efficient,Childs2011,Childs2012} based on Lie product formulas achieved scaling like $\mathcal{O}(\op{poly}(t,d,\|H\|,1/\epsilon))$. Steady improvements based on alternate techniques such as quantum walks~\cite{Childs2010}, fractional queries~\cite{Berry2014}, linear-combination-of-unitaries~\cite{Berry2015Hamiltonian,Berry2016corrected}, and quantum signal processing~\cite{Low2016}, have recently culminated in an algorithm~\cite{Low2016HamSim} with optimal scaling $\mathcal{O}(td\|H\|_\mathrm{max}+\log{(1/\epsilon)})$\footnote{For readability, we use this instead of the precise complexity~$\Theta(t d\|H\|_\mathrm{max}+\frac{\log{(1/\epsilon)}}{\log{(e+\log{(1/\epsilon)}/(td\|H\|_\mathrm{max}))}})$~\cite{Gilyen2018quantum}.} with respect to all parameters. Nevertheless, the possibility of further improvement is highlighted by~\cite{Low2017USA} which uses uniform spectral amplification to obtain a more general result scaling like $\tilde{\mathcal{O}}(t\sqrt{d\|H\|_\mathrm{max}\|H\|_1}\log{(1/\epsilon)})$, and~\cite{Low2018IntPicSim} which uses the interaction picture to obtain logarithmic scaling with respect to the diagonal component of $H$. 

Depending on the application, it is more natural to express complexity in terms one choice of parameters over another. To illustrate, Hamiltonians with unit spectral norm $\|H\|=1$ are simulated as a subroutine in quantum algorithms that solve systems of linear equations~\cite{Childs2015LinearSystems}. In particular, the algorithm that is optimal in only $t,d$, and max-norm $\|H\|_\mathrm{max} = \max_{ik}|H_{ik}|$ is unable to exploit this prior information. Though the spectral norm is a natural parameter in many problems of interest,~\cite{Childs2010Limitation} have ruled out sparse simulation algorithms scaling like $\mathcal{O}(\op{poly}(\|H\|)t)$. To date, the best simulation algorithm that exploits knowledge of $\|H\|$ has query complexity $\tilde{\mathcal{O}}((t\|H\|)^{4/3}d^{2/3}/\epsilon^{1/3})$~\cite{Berry2012}, but the best lower bound, based on quantum search, is $\Omega(\sqrt{d})$. 

While some recent developments~\cite{Wang2018NonSparse,Chakraborty2018BlockEncoding} achieve better scaling like $\tilde{\mathcal{O}}(t\sqrt{N}\|H\|)$ for non-sparse Hamiltonians, this is in terms of queries to a very strong quantum-RAM~\cite{Giovannetti2009qRAM} oracle. This model is incomparable to standard black-box queries as it enables the ability to prepare arbitrary quantum states with $\mathcal{O}(\log{(N)})$ queries, which is at odds with the standard $\Omega(\sqrt{N})$ quantum search query lower bound. Thus identifying the optimal trade-off between $t,d,\|H\|$ in the black-box setting remains a fundamental open problem with useful applications. 

We present an algorithm for simulating sparse Hamiltonians with a complexity trade-off between time, sparsity, and spectral norm, that is optimal up to subpolynomial factors. The query complexity of our algorithm is as follows.
\begin{restatable}[Sparse Hamiltonian simulation with dependence on $\|H\|_{1 \shortto 2}$]{theorem}{HamSimSubordinate}
	\label{thm:HamSim_sparse}
	Let	$H\in \mathbb{C}^{N\times N}$ be a $d$-sparse Hamiltonian satisfying all the following:
	\begin{itemize}
		\item There exist oracles of~\cref{def:Sparse_Oracle} for the positions of nonzero matrix entries, and their values to $b$ bits of precision.
		\item An upper bound on the subordinate norm $\Lambda_{1 \shortto 2}\ge \|H\|_{1 \shortto 2} = \max_{k}\sqrt{\sum_i|H_{ik}|^2}$ is known.
	\end{itemize}
	Let $\tau=t\sqrt{d}\Lambda_{1 \shortto 2}$. Then the time-evolution operator $e^{-iHt}$ may be approximated to error $\epsilon$ for any $t>0$ with query complexity $C_\mathrm{queries}[H,t,\epsilon]$ and $C_\mathrm{gates}[H,t,\epsilon]$ additional arbitrary two-qubit gates where
	\begin{align}
	C_\mathrm{queries}[H,t,\epsilon]&=\mathcal{O}\left(\tau\lr{\log{\lr{\tau/\epsilon}}}^{\mathcal{O}(\sqrt{\log{d}})}\right)=\mathcal{O}\left(\tau\left(\tau/\epsilon\right)^{o(1)}\right),
	\\\nonumber
	C_\mathrm{gates}[H,t,\epsilon]&=\mathcal{O}\left(C_\mathrm{queries}[H,t,\epsilon](\log{(N)}+b)\right).
	\end{align}
\end{restatable}
The scaling of our algorithm with the subordinate norm $\|H\|_{1 \shortto 2}$ rather than the spectral norm $\|H\|$ is in fact a stronger result than expected. In the worst-case, $\|H\|_{1 \shortto 2}\le\|H\|$, and in the best-case, $\|H\|_{1 \shortto 2}\ge \frac{\|H\|}{\sqrt{d}}$~\cite{Childs2010Limitation}. Indeed, any algorithm scaling like $\mathcal{O}(t\sqrt{d}\|H\|/\plog{(td\|H\|)})$ would violate a quantum search lower bound~\cite{Berry2012}. Up to subpolynomial factors,~\cref{thm:HamSim_sparse} provides a strict improvement over all prior sparse Hamiltonian simulation algorithms~\cite{Low2016HamSim,Low2017USA}, as seen by  substituting the inequalities $\sqrt{d}\|H\|_{1 \shortto 2}\le \sqrt{d\|H\|_\mathrm{max}\|H\|_1}\le d\|H\|_\mathrm{max}$. The linear scaling of gate complexity $\mathcal{O}(b)$ with respect to the bits of $H$, where $b$ is independent of $\tau/\epsilon$, is notable. Previous approaches scale with $\mathcal{O}(b^{5/2})$~\cite{Berry2015Hamiltonian} in addition to requiring scaling $b=\mathcal{O}(\log{(\tau/\epsilon)})$. 

 We also prove optimality through a lower bound. This lower bound is based on finding a Hamiltonian with a known query complexity, that also allows for the independent variation of these parameters as follows.
\begin{restatable}[Lower bound on sparse Hamiltonian simulation.]{theorem}{HamSimLowerBound}
	\label{thm:Lower_bound}
	For real numbers $d>1,s>1,t>0$, there exists a $d$-sparse Hamiltonian $H$ with $\|H\|_{1 \shortto 2}\le s$ such that the query complexity of simulating time-evolution $e^{-iHt}$ with bounded error is
	\begin{align}
\Omega(t\sqrt{d}\|H\|_{1 \shortto 2}).
	\end{align} 
\end{restatable}

This simulation algorithm solves an open problem on the query complexity of black-box unitary implementation~\cite{Jordan2009,Berry2012}. Whereas Hamiltonian simulation is concerned with approximating a unitary $e^{-iHt}$ where given a description of $H$, the black-box unitary problem is concerned with approximating a unitary $U$ where $U$ itself is directly described by black-box oracles. Similar to the simulation problem, the decision variant of black-box unitary implementation is $\mathsf{BQP}$-complete. Though the query complexity of previous algorithms for this problem is $\tilde{\mathcal{O}}(d^{2/3}/\epsilon^{1/3})$, it was conjectured by~\cite{Berry2012} that the optimal scaling is closer to $\mathcal{O}(\sqrt{d})$. This is motivated by an $\Omega(\sqrt{d})$ lower bound for the case where $U$ solves a search problem. Using a simple reduction by~\cite{Jordan2009} to the Hamiltonian simulation problem, we show that this lower bound is tight up to subpolynomial factors as follows.
\begin{restatable}[Query complexity of black-box unitary implementation]{corollary}{BlackBoxUnitary}
	\label{thm:BlackBoxUnitary}
	Let	$U\in \mathbb{C}^{N\times N}$ be a $d$-sparse unitary satisfying all the following:
	\begin{itemize}
		\item There exist oracles of~\cref{def:Sparse_Oracle} for the positions and values of non-zero matrix elements.
	\end{itemize}
	Then the $U$ may be approximated to error $\epsilon$ with query complexity
	\begin{align}
	C_\mathrm{queries}=\mathcal{O}\left(\sqrt{d}\lr{\log{\lr{d/\epsilon}}}^{\mathcal{O}(\sqrt{\log{d}})}\right)
	=\mathcal{O}\left(\sqrt{d}\left(d/\epsilon\right)^{o(1)}\right).
	\end{align}
\end{restatable}

Solving systems of linear equations is another application of this algorithm. The original black-box formulation of this problem~\cite{Harrow2009}, which also has $\mathsf{BQP}$-complete decision variant, specifies a $d$-sparse matrix $A\in\mathbb{C}^{N\times N}$ with condition number $\kappa$ to be described by black-box oracles, and a state $\ket{b}\in\mathbb{C}^{N}$. The goal then is to approximate the state $A^{-1}\ket{b}/\|A^{-1}\ket{b}\|$ to error $\epsilon$. Though previous algorithms achieve this using $\mathcal{O}(\kappa d\plog{(\kappa d/\epsilon)})$ queries~\cite{Childs2015LinearSystems}, this falls short of the lower bound $\Omega(\kappa \sqrt{d}\log{(1/\epsilon)})$~\cite{Harrow2018}. By invoking quantum linear system solvers based on the block-encoding framework~\cite{Chakraborty2018BlockEncoding}, we match this lower bound up to subpolynomial factors.
\begin{restatable}[Query complexity of solving sparse systems of linear equations]{corollary}{BlackBoxQLSP}
	\label{thm:BlackBoxQLSP}
	Let	$A\in \mathbb{C}^{N\times N}$ be a $d$-sparse matrix satisfying all the following:
	\begin{itemize}
		\item There exist oracles of~\cref{def:Sparse_Oracle} for the positions and values of non-zero matrix elements.
		\item The spectral norm $\|A\|= 1$ and the condition number $\|A^{-1}\|\le\kappa$.
	\end{itemize}
	Let $\ket{b}\in\mathbb{C}^{N}$ be prepared by a unitary oracle $O_{b}\ket{0}=\ket{b}$. 
	Then the query complexity to all oracles for preparing a state $\ket{\psi}$ satisfying $\left\|\ket{\psi}-\frac{A^{-1}\ket{b}}{\|A^{-1}\ket{b}\|}\right\|\le\epsilon$ is
	\begin{align}
	C_\mathrm{queries}	=\mathcal{O}\left(\kappa\sqrt{d}\left(\kappa d/\epsilon\right)^{o(1)}\right).
	\end{align}
\end{restatable}

A key result in obtaining~\cref{thm:HamSim_sparse} is a general-purpose simulation algorithm for Hamiltonians $H=\sum^m_{j=1}H_j$ described by a sum of Hermitian terms, and is of independent interest. Roughly speaking, let $\alpha_j\ge \|H_j\|$ bound the spectral norm of each term, and let $C_j$ be the cost of simulating each term $H_j$ alone for constant time $t\alpha_j =\mathcal{O}(1)$. Then previous algorithms~\cite{Berry2015Truncated,Low2016hamiltonian} for simulating the full $H$ have cost scaling like
$\mathcal{O}\left(t\|\vec{\alpha}\|_1\|\vec{C}\|_1\right)$, which combines the worst-cases of $\alpha_j$ and $C_j$. In contrast, we describe in~\cref{Thm:HamSimRecursion} a simulation algorithm that scales like $\mathcal{O}(t\langle\vec{\alpha},\vec{C}\rangle e^{\mathcal{O}(m)})$, but picks up an exponential prefactor. This algorithm, which extends work by~\cite{Low2018IntPicSim}, is advantageous when $m$ held constant, and the cost of each term scales like $C_j=\mathcal{O}(1/\alpha_j)$. Though this condition appears artificial, it is situationally useful. 

At the highest-level, our main simulation result~\cref{thm:HamSim_sparse} involves three steps. The first step, similar to~\cite{Berry2012}, splits a sparse Hamiltonian into a sum of $m$ terms $H_j$, where the $j^{\text{th}}$ term contains all matrix entries of $H$ with absolute value thresholded between $(\Lambda^{(j-1)}_\mathrm{max},\Lambda^{(j)}_\mathrm{max}]$. The second step uses a modification of the uniform spectral amplification technique by~\cite{Low2017USA} and an upper bound on the spectral norm $\|H_j\|\le\|H_j\|_1\le \alpha_j= \Lambda^2_{1 \shortto 2}/\Lambda^{(j-1)}_\mathrm{max}$ to simulate each term with cost $C_j=\mathcal{O}((d\Lambda^{(j)}_\mathrm{max}/\alpha_j)^{1/2})$ --  a different bound $\alpha_1$ is used for the $j=1$ term. Finally, we recombine these terms using the general-purpose simulation algorithm of~\cref{Thm:HamSimRecursion}. The stated result is obtained by a judicious choice of these thresholds $\Lambda^{(j)}_\mathrm{max}$ and optimizing for $m$ as a function of $d$.

In~\cref{sec:overview} we provide a more detailed overview of our algorithms.~\cref{sec:Ham_sim_recursion} derives the general-purpose simulation algorithm~\cref{Thm:HamSimRecursion}, which is obtained from a recursive application of Hamiltonian simulation in the interaction picture by~\cite{Low2018IntPicSim}, outlined in~\cref{sec:IntPicSim}.~\cref{sec:sparse_ham_sim} derives the sparse Hamiltonian simulation algorithm~\cref{thm:HamSim_sparse}, which applies uniform spectral amplification by~\cite{Low2017USA}, outlined in~\cref{sec:USA}, and modified in~\cref{sec:arithmetic-free} to obtain an improved precision scaling.~\cref{Sec:Lower_Bound} proves the lower bound~\cref{thm:Lower_bound} on sparse Hamiltonian simulation. The example applications are found in~\cref{sec:blackboxunitary}, which derives the result on black-box unitary implementation~\cref{thm:BlackBoxUnitary}, and~\cref{sec:blackboxQLSP}, for the result on solving sparse systems of linear equations~\cref{thm:BlackBoxQLSP}. We conclude in~\cref{sec:conclusion}.

\section{Overview of algorithms}
\label{sec:overview}
Recent simulation algorithms~\cite{Low2017USA,Low2018IntPicSim} have focused on simulating Hamiltonians described by so-called `standard-form' oracles~\cite{Low2016hamiltonian} (more descriptively called `block-encoding' in~\cite{Chakraborty2018BlockEncoding}). There, it is assumed that $H$, in some basis, is embedded in the top-left block of a unitary oracle as follows.
\begin{restatable}[Block-encoding framework]{definition}{Blockencoding}
	\label{Def:Standard_Form}
	A matrix ${H} \in \mathbb {C}^{N_s\times N_s}$ that acts on register $s$ is block encoded by any unitary $U$ where the top-left block of $U$, in a known computational basis state $\ket{0}_a\in\mathbb{C}^{N_a}$ on register $a$, is equal to $H/\alpha$ for some normalizing constant $\alpha \ge \|H\|$:
	\begin{align}
	U = 
	\left(\begin{matrix}
	H/\alpha & \cdot \\
	\cdot & \cdot
	\end{matrix}\right),
	\quad
	(\bra{0}_a\otimes \ii_s)U(\ket{0}_a\otimes \ii_s) = \frac{H}{\alpha}.
	\end{align}
\end{restatable}
Now given a Hamiltonian expressed as a sum of $m$ Hermitian terms
\begin{align}
H=\sum^m_{j=1}H_j,
\end{align} 
where each term $H_j$ is assumed to be block-encoded by a unitary oracle $U_j$, the cost $C_\mathrm{total}[H,t,\epsilon]$ of Hamiltonian simulation to error $\epsilon$ then expressed in terms of $\alpha_j$ and number of times each $U_j$ is applied, and any additional arbitrary two-qubit quantum gates. As each $U_j$ might differ in complexity, we assign them a cost $C_j$, and also assume that the cost of a controlled-$U_j$ operator is $\mathcal{O}(C_j)$. 
Depending on context, which should be clear, $C_j$ could refer to a query or gate complexity. We also find it useful to distinguish between two types of costs as follows.
\begin{itemize}
	\item $C_\mathrm{queries}[H,t,\epsilon]=\sum_{j=1}M_j C_j$, where $M_j$ is the number of queries made to $U_j$.
	\item $C_\mathrm{gates}[H,t,\epsilon]$ is the number of any additional arbitrary two-qubit quantum gates required. 
\end{itemize}
In other words, if $C_j$ is a measure of gate complexity, then the total cost of simulation is
\begin{align}
C_\mathrm{total}[H,t,\epsilon]=C_\mathrm{queries}[H,t,\epsilon]+C_\mathrm{gates}[H,t,\epsilon].
\end{align}
In general, these queries are expensive, and so cost is dominated by $C_\mathrm{queries}[H,t,\epsilon]$. In other cases, particularly in sparse Hamiltonian simulation, $C_j$ is the number of queries made to the more fundamental oracles described in~\cref{def:Sparse_Oracle}.

Error is commonly defined as follows~\cite{Berry2015Truncated,Childs2017Speedup}. Any quantum circuit $U\in\mathbb{C}^{N_sN_a\times N_sN_a}$ that approximates a unitary operator $A\in\mathbb{C}^{N_s\times N_s}$, say time-evolution $A= e^{-iHt}$, to error $\epsilon$ is a block-encoding of $A$ such that $\|(\bra{0}_a\otimes I_s)U(\ket{0}_a\otimes I_s)-A\|\le\epsilon$. The error of approximating a product of $A_1,\cdots,A_m$ where each is approximated by a unitary $U_j$ is then $\|(\bra{0}_a\otimes I_s)U_m(\ket{0}_a\otimes I_s)\cdots(\bra{0}_a\otimes I_s) U_1(\ket{0}_a\otimes I_s) - A_m\cdots A_1\|\le m\epsilon$, with failure probability $\mathcal{O}(m\epsilon)$. This is useful when obtaining longer time-evolution by concatenating shorter time-evolution. Alternatively, if we choose to not project onto the $\ket{0}_a$ state, this implies an error $\max_{\|\ket{\psi}\|=1}\|[U-I_a\otimes A]\ket{0}\ket{\psi}\|\le\epsilon+\sqrt{2\epsilon(1-\epsilon)}$ for a single operator, and $\max_{\|\ket{\psi}\|=1}\|[U_m\cdots U_1 - I\otimes (A_m\cdots A_1)]\ket{0}\ket{\psi}\|\le m(\epsilon+\sqrt{2\epsilon(1-\epsilon)})$ for a sequence. Our results are insensitive to either definition of error as both cases scale linearly with $m$, and all factors of $1/\epsilon$ later on occur in subpolynomial factors. 

With existing algorithms~\cite{Berry2015Truncated,Low2016hamiltonian}, the cost of simulation is
\begin{align}
\label{eq:qubitization}
C_\mathrm{queries}[H,t,\epsilon]&=\mathcal{O}\left(\left(t\|\vec{\alpha}\|_1+\log{(1/\epsilon)}\right)\|\vec{C}\|_1\right),\quad \|\vec{\alpha}\|_1=\sum^m_{j=1}\alpha_j,\quad \|\vec{C}\|_1=\sum^m_{j=1}C_j,
\\\nonumber
C_\mathrm{gates}[H,t,\epsilon]&=\mathcal{O}\left(\left(t\|\vec{\alpha}\|_1+\log{(1/\epsilon)}\right)\log{(N_a)}\right),
\end{align}
which combines the worst-cases of $\alpha_j$ and $C_j$ separately. The cost of our algorithm~\cref{Thm:HamSimRecursion} instead scales with the sum $\sum^m_{j=1}\alpha_j C_j$, but picks up an exponential prefactor $e^{\mathcal{O}(m)}$ as follows.
\begin{restatable}[Hamiltonian simulation by recursion in the interaction picture]{theorem}{HamSimRecursion}
	\label{Thm:HamSimRecursion}
	For any Hamiltonian $H\in\mathbb{C}^{N\times N}$,  let us assume that
	\begin{itemize}
		\item $H=\sum^m_{j=1}H_j$ is a sum of $m$ Hermitian terms.
		\item Block-encoding $H_j/\alpha_j$ with ancilla dimension $N_a$ has cost $C_j$.
		\item The normalizing constants are sorted like $\alpha_1 \ge \alpha_2 \ge\cdots \ge \alpha_m > 0$.
	\end{itemize}
	Then the time-evolution operator $e^{-iHt}$ may be approximated to error $\epsilon$ for any $t>0$ with cost	
	\begin{align}
	\label{eq:costRecursion}
	C_\mathrm{queries}[H,t,\epsilon]&=
	\mathcal{O}\left(t\langle\vec{\alpha},\vec{C}\rangle\log^{2m-1}{\lr{\frac{t\alpha_{1}}{\epsilon}}} \right),\quad \langle\vec{a},\vec{C}\rangle=\sum^m_{j=1}\alpha_j C_j,
	\\\nonumber
	C_\mathrm{gates}[H,t,\epsilon]&=
	\mathcal{O}\left(t\|\vec{\alpha}\|_1\log^{2m-1}{\lr{\frac{t\alpha_{1}}{\epsilon}}} \log{(N_a)}\right).
	\end{align}
\end{restatable}
This is achieved by repeatedly applying the interaction picture Hamiltonian simulation algorithm of~\cite{Low2018IntPicSim}, which we briefly outline. For simplicity, let us ignore error contributions, as they occur in polylogarithmic factors. Given a Hamiltonian $H=A+B$ with two terms, let us assume that $B/\alpha_B$ is block-encoded with cost $C_B$, and that $e^{-iAt}$ may be simulated with cost $\tilde{\mathcal{O}}(t\alpha_A C_A)$ for some cost$C_A$ and $\alpha_A\ge\|A\|$.~\cite{Low2018IntPicSim} then simulates $e^{-i(A+B)t}$ with cost $\tilde{\mathcal{O}}(t(\alpha_BC_B+\alpha_AC_A))$. Given a Hamiltonian described by a sum of $m$ block-encoded terms $H_j$, the proof of~\cref{Thm:HamSimRecursion} then follows by induction. For $k=1$, one simulates $e^{-iH_1t}$ using, say, an algorithm with the complexity of~\cref{eq:qubitization}. For $k>1$, one simulates $e^{-i(H_1+\cdots+H_k)t}$ by combining $e^{-i(H_1+\cdots+H_{k-1})t}$ with the block-encoding of $H_k$. One repeats until $k=m$, and each step contributes a multiplicative factor.

Given a $d$-sparse Hamiltonian described by the black-box oracles of~\cref{def:Sparse_Oracle}, we split it into $m$ terms $H_j$ where the $j^{\text{th}}$ term only contains entries of $H$ with absolute value between $(\Lambda^{(j-1)}_\mathrm{max},\Lambda^{(j)}_\mathrm{max}]$. Subsequently, we block-encode each term in the format of~\cref{Def:Standard_Form}. This block-encoding requires some number $C_j$ of queries, and achieves some normalizing constant $\alpha_j$. The total cost of simulation given by~\cref{eq:costRecursion} depends crucially on the quality of this encoding -- clearly, we would like to minimize $C_j$ and $\alpha_j$ such that for any fixed $m$,  all products $\alpha_jC_j$ scale identically with respect to $d$.

Block-encoding a sparse Hamiltonian is related to the Szegedy quantum walk defined for any Hamiltonian~\cite{Childs2010}. In the $m=1$ case, one defines a set of quantum states $\ket{\chi_k}$ and $\ket{\bar{\chi}_i}$, such that they have $\mathcal{O}(d)$ nonzero amplitudes at known positions and their mutual overlap is an amplitude $\braket{\bar{\chi}_i}{\chi_k}=H_{ik}/\alpha$ that reproduces matrix values of $H$, up to a normalizing constant. By doubling the Hilbert space of $H$ so that $N_a=N$ and using two additional qubits, all states within each set $\{\ket{\chi_k}\}$, and $\{\ket{\bar\chi_i}\}$ can be made mutually orthogonal. Block-encoding then reduces to controlled-arbitary quantum state preparation where on an input state $\ket{k}_s$, one prepares a quantum state $\ket{\chi_k}$, then unprepares with a similar procedure for $\ket{\bar{\chi}_i}$. 

A key step in controlled-state preparation applied in~\cite{Berry2012,Berry2015Hamiltonian,Low2017USA}, is converting a $b$-bit binary representation of a matrix element $H_{ik}$ in quantum state $\ket{H_{ik}}\ket{0}_a$ to an amplitude $\ket{H_{ik}}\left(\sqrt{\frac{H_{ik}}{\Lambda_\mathrm{max}}}\ket{0}_a+\cdots\ket{1}_a\right)$. Previous approaches require coherently computing a binary representation of $\ket{H_{ik}}\mapsto\ket{\sin^{-1}(\sqrt{H_{ik}/\Lambda_{\mathrm{max}}})}$. This is performed with quantum arithmetic and is extremely costly, both in the asymptotic limit and in constant prefactors. Moreover, the function can only be approximated, leading to the number of bits $b$ scaling with error. Using recent insight on arithmetic-free black-box quantum state preparation~\cite{Sanders2018}, this subroutine may be replaced with an easier problem of preparing the desired amplitude with a garbage state $\ket{u_{|H_{ik}|}}_b$, that depends only on $|H_{ik}|$, attached like $\ket{H_{ik}}\ket{0}_b\ket{0}_a\mapsto\ket{H_{ik}}\left(\sqrt{\frac{H_{ik}}{\Lambda_\mathrm{max}}}\ket{u_{|H_{ik}|}}_b\ket{0}_a+\cdots\right)$. This suffices to implement the walk, and as shown in the proof of~\cref{thm:sparse-Ham-block-encoding}, can be performed exactly with $O(1)$ reversible integer adders, hence the linear $\mathcal{O}(b)$ gate complexity scaling.

As shown by~\cite{Berry2012}, one may then block-encode $H_j$ with $\alpha_j=d\Lambda^{(j)}_\mathrm{max}$ using $C_j=\mathcal{O}(1)$.
By enhancing quantum state preparation with a linearized high-precision variant of amplitude amplification,~\cite{Low2017USA} encodes $H$ to error $\epsilon$ with $\alpha_j=\Theta(\Lambda^{(j)}_1)$ using $C_j=\tilde{\mathcal{O}}((d\Lambda^{(j)}_\mathrm{max}/\alpha_j)^{1/2})$ queries. By bounding the induced one-norm $\|H_j\|_1\le \Lambda^{(j)}_1\le \Lambda^2_{1 \shortto 2}/\Lambda^{(j-1)}_\mathrm{max}$, the overall cost of simulation by~\cref{eq:costRecursion} scales with the sum of $\alpha_jC_j=\tilde{\mathcal{O}}((d\Lambda^{(j)}_\mathrm{max}/\Lambda^{(j-1)}_\mathrm{max})^{1/2}\Lambda_{1 \shortto 2})$ (the $j=1$ case uses a different bound $\Lambda^{(1)}_1 \le \sqrt{d}\Lambda_{1 \shortto 2}$). For any fixed $m$, we choose the ratio between absolute values of terms to scale like $\Lambda^{(j)}_\mathrm{max}/\Lambda^{(j-1)}_\mathrm{max}=d^{\mathcal{O}(1/m)}$. A straightforward optimization of the exponents leads to the stated complexity of~\cref{thm:HamSim_sparse} with $m=\mathcal{O}(\sqrt{\log{(d)}})$.

\section{Hamiltonian simulation by recursion}
\label{sec:Ham_sim_recursion}
In this section, we prove~\cref{Thm:HamSimRecursion} for simulating time-evolution by Hamiltonians $H$ expressed as a sum of $m$ Hermitian terms
\begin{align}
H=\sum^m_{j=1}H_j.
\end{align} 
Typically, time-evolution $e^{-iH_jt}$ by each term alone is easy to implement -- the challenge is combining these parts to approximate time-evolution $e^{-iHt}$ by the whole. In our algorithm, each Hamiltonian $H_j$ is assumed to be block-encoded by a unitary oracle in the format of~\cref{Def:Standard_Form}.
\begin{proof}[Proof of~\cref{Thm:HamSimRecursion}.]
We combine two other simulation algorithms by~\cite{Berry2015Truncated,Low2016hamiltonian,Low2018IntPicSim} in an $m$ step recursive procedure. At the $k=1$ step, the first algorithm~\cref{thm:sim_block} approximates $e^{-iH_1t}$ to error $\epsilon$ for time $t>0$ using the block-encoding of $H_1$. At the $k>1$ step, the second algorithm~\cref{thm:int_pic_sim} approximates $e^{-i(H_1+\cdots+H_k)t}$ by combining $e^{-i(H_1+\cdots+H_{k-1})t}$ with the block-encoding of $H_k$. By repeating this step for $k=2,\cdots,m$, we obtain the time-evolution operator $e^{-iHt}$. We now state these simulation algorithms. As~\cref{thm:int_pic_sim} modifies the original presentation of the same result in~\cite{Low2018IntPicSim}, we provide a proof sketch~\cref{sec:IntPicSim}.
\begin{restatable}[Hamiltonian simulation of a single term~\cite{Berry2015Truncated,Low2016hamiltonian}]{lemma}{HamSimBlock}
	\label{thm:sim_block}
	For any Hamiltonian $A\in\mathbb{C}^{N\times N}$,  let us assume that
	\begin{itemize}
		\item Block-encoding $A/\alpha$ with ancilla dimension $N_a$ has cost $C_A$. 
	\end{itemize}
	Then the time-evolution operator $e^{-iHt}$ may be approximated to error $\epsilon$ for any $t>0$ with cost
	\begin{align}
	C_\mathrm{queries}[H,t,\epsilon]=\mathcal{O}\left(\left(t\alpha+1\right)C_A\log{(t\alpha /\epsilon)}\right),
	\;
	C_\mathrm{gates}[H,t,\epsilon]=\mathcal{O}\left(\left(t\alpha+1\right)\log{(t\alpha/\epsilon)}\log{(N_a)}\right).
	\end{align}
\end{restatable}
\begin{restatable}[Hamiltonian simulation in the interaction picture; adapted from~\cite{Low2018IntPicSim}]{lemma}{HamSimIntPic}
	\label{thm:int_pic_sim}
	For any Hamiltonian $A+B\in\mathbb{C}^{N\times N}$, where $A$ and $B$ are Hermitian, let us assume that
	\begin{itemize}
		\item Block-encoding $B/\alpha_B$ with ancilla dimension $N_a$ has cost $C_B$. 
		\item There exists some $\alpha_A\ge\alpha_B$, $\gamma>0$, such that approximating $e^{-iAt}$ to error $\epsilon$ for any time $t>0$ has cost
		\begin{align}
		C_\mathrm{queries}[A,t,\epsilon]&=\mathcal{O}\left((t\alpha_A+1)\log^{\gamma}\left(t \alpha_A/\epsilon\right)\right),
		\\\nonumber
		C_\mathrm{gates}[A,t,\epsilon]&=\mathcal{O}\left((t\alpha_A+1)\log^{\gamma}\left(t\alpha_A/\epsilon\right)\log{(N_a)}\right).
		\end{align}
	\end{itemize}
	Then the time-evolution operator $e^{-iHt}$ may be approximated to error $\epsilon$ for any $t>0$ with cost
	\begin{align}
	\label{eq:int_pic_sim_cost}
	C_\mathrm{queries}[B,A,t,\epsilon]&=\mathcal{O}\left((t\alpha_B+1)
	\left(C_B+C_\mathrm{queries}\left[A,\frac{1}{\alpha_B},\frac{\epsilon}{t\alpha_B}\right]\right)\log^2{\lr{\frac{t\alpha_A}{\epsilon}}}
	\right),
	\\\nonumber
	C_\mathrm{gates}[B,A,t,\epsilon]&=\mathcal{O}\left((t\alpha_B+1)
	C_\mathrm{gates}\left[A,\frac{1}{\alpha_B},\frac{\epsilon}{t\alpha_B}\right]\log^2{\lr{\frac{t\alpha_A}{\epsilon}}}
	\right).
	\end{align}
\end{restatable}

Note the factor $\mathcal{O}(t\alpha+1)$, which provides the correct dominant scaling in the nonasymptotic limit where $t \alpha=\mathcal{O}(1)$. Let us apply~\cref{thm:int_pic_sim} recursively with $A=H_{< k}=\sum^{k-1}_{j=1}H_j$ and $B=H_k$ for $k>1$. The cost of the $k^{\text{th}}$ iteration depends on the cost of the $k-1^{\text{th}}$ iteration as follows:
\begin{align}
C_\mathrm{queries}[H_1,\cdot,t,\epsilon]&=\mathcal{O}\left((t\alpha_1+1)C_1\log\left(t\alpha_1/\epsilon\right)\right)
\displaybreak[0]\\\nonumber
C_\mathrm{queries}[H_2,H_{<2},t,\epsilon]&=\mathcal{O}\left((t\alpha_2+1)
\left(C_2+C_\mathrm{queries}\left[H_1,\cdot,\frac{1}{\alpha_2},\frac{\epsilon}{t\alpha_2}\right]\right)\log^2{\lr{\frac{t\alpha_1}{\epsilon}}}
\right)
\\\nonumber
&=\mathcal{O}\left((t\alpha_2+1)
\left(C_2+\lr{\frac{\alpha_1}{\alpha_2}+1}C_1\right)\log^3{\lr{\frac{t\alpha_1}{\epsilon}}}
\right)
\\\nonumber
&=\mathcal{O}\left((t\alpha_2+1)
\left(C_2+2\frac{\alpha_1}{\alpha_2}C_1\right)\log^3{\lr{\frac{t\alpha_1}{\epsilon}}}
\right)
\displaybreak[0]\\\nonumber
C_\mathrm{queries}[H_3,H_{<3},t,\epsilon]&=\mathcal{O}\left((t\alpha_3+1)
\left(C_3+C_\mathrm{queries}\left[H_2,\cdot,\frac{1}{\alpha_3},\frac{\epsilon}{t\alpha_3}\right]\right)\log^2{\lr{\frac{t\alpha_1}{\epsilon}}}
\right),
\\\nonumber
&=\mathcal{O}\left((t\alpha_3+1)
\left(C_3+2\frac{\alpha_2}{\alpha_3}
\left(C_2+2\frac{\alpha_1}{\alpha_2}C_1\right)\right)\log^5{\lr{\frac{t\alpha_1}{\epsilon}}}
\right),
\displaybreak[0]\\\nonumber
&\vdots
\displaybreak[0]\\\nonumber
C_\mathrm{queries}[H_k,H_{< k}, t,\epsilon]
&=\mathcal{O}\left((t\alpha_k+1)\left(\sum^k_{j=1}C_j\prod^{k}_{i=j+1}\frac{2\alpha_{i-1}}{\alpha_{i}}\right)\log^{2k-1}{\lr{\frac{t\alpha_1 }{\epsilon}}}\right)
\\\nonumber
&=\mathcal{O}\left(t\left(\sum^k_{j=1}\alpha_jC_j\right)\left(2\log{\lr{\frac{t\alpha_1 }{\epsilon}}}\right)^{2k-1}\right).
\\\nonumber
\end{align}
The gate complexity follows by an identical recursion:
\begin{align}
C_\mathrm{gates}[H_1,\cdot,t,\epsilon]&=\mathcal{O}\left((t\alpha_1+1)\log\left(t\alpha_1/\epsilon\right)\log{(N_a)}\right),
\\\nonumber
C_\mathrm{gates}[H_2,H_{<2},t,\epsilon]&=\mathcal{O}\left((t\alpha_2+1)
\left(C_\mathrm{gates}\left[H_1,\cdot,\frac{1}{\alpha_2},\frac{\epsilon}{t\alpha_2}\right]\right)\log^2{\lr{\frac{t\alpha_1}{\epsilon}}}
\right)
\\\nonumber
&=\mathcal{O}\left((t\alpha_2+1)
2\frac{\alpha_1}{\alpha_2}\log^3{\lr{\frac{t\alpha_1}{\epsilon}}}\log{(N_a)}\right),
\\\nonumber
&\vdots
\\\nonumber
C_\mathrm{gates}[H_k,H_{< k}, t,\epsilon]
&=\mathcal{O}\left((t\alpha_k+1)\left(\prod^{k}_{i=j+1}\frac{2\alpha_{i-1}}{\alpha_{i}}\right)\log^{2k-1}{\lr{\frac{t\alpha_1 }{\epsilon}}}\log{(N_a)}\right)
\\\nonumber
&=\mathcal{O}\left(t\left(\sum^k_{j=1}\alpha_j\right)\left(2\log{\lr{\frac{t\alpha_1 }{\epsilon}}}\right)^{2k-1}\log{(N_a)}\right).
\\\nonumber
\end{align}
Note that we use $\frac{\alpha_{j-1}}{\alpha_j}\ge 1$ to simplify $\frac{\alpha_{j-1}}{\alpha_j}+1\le 2\frac{\alpha_{j-1}}{\alpha_j}$. Setting $k=m$ completes the proof.
\end{proof}

\section{Sparse Hamiltonian simulation with $t\sqrt{d}\|H\|_{1\rightarrow 2}$ scaling}
\label{sec:sparse_ham_sim}
We now simulate sparse Hamiltonians by applying the algorithm of~\cref{sec:Ham_sim_recursion}. In the standard definition, a Hamiltonian is $d$-sparse if it has at most $d$-nonzero entries in any row. Moreover, it is assumed that there exists black-box oracles to compute the positions and $b$-bit values of these entries. The cost of simulation is the number of queries made to these oracles, which are more precisely defined by~\cref{def:Sparse_Oracle}. The query complexity of block-encoding a sparse Hamiltonian $H$ is given by the following result, which was mostly proven by~\cite{Low2017USA}. Our contribution is improving the gate complexity scaling from $\mathcal{O}(b^{5/2})$ to $\mathcal{O}(b)$.
\begin{restatable}[Block encoding of sparse Hamiltonians by amplitude multiplication; modified from~\cite{Low2017USA}]{theorem}{USAblockencoding}
	\label{thm:sparse-Ham-block-encoding}
	Let	$H\in \mathbb{C}^{N\times N}$ be a $d$-sparse Hamiltonian satisfying all the following:
	\begin{itemize}
		\item There exist oracles $O_F$ and $O_H$ of~\cref{def:Sparse_Oracle} that compute the positions and values of non-zero matrix elements to $b$ bits of precision.
		\item An upper bound on the max-norm $\Lambda_{\mathrm{max}}\ge \|H\|_\mathrm{max} = \max_{ik}|H_{ik}|$ is known.
		\item An upper bound on the spectral-norm $\Lambda\ge \|H\| = \max_{v\neq 0}\frac{\|H\cdot v\|}{\|v\|}$ is known. 
		\item An upper bound on the induced one-norm $\Lambda_{1}\ge \|H\|_{1} = \max_{k}\sum_i|H_{ik}|$ is known.
	\end{itemize}
	Then there exists a Hamiltonian $\tilde H \in \mathbb{C}^{N\times N}$ that approximates $H$ with error $\|\tilde H- H\|=\mathcal{O}(\Lambda\delta)$, and can be block-encoded with normalizing constant $\alpha=\Theta(\Lambda_{1})$ using
	\begin{itemize}
		\item Queries $O_F$ and $O_H$: $\mathcal{O}\left(\sqrt{\frac{d\Lambda_{\mathrm{max}}}{\Lambda_{1}}}\log{\lr{\frac{1}{\delta}}}\right)$.
		\item Quantum gates: $\mathcal{O}\left(\sqrt{\frac{d\Lambda_{\mathrm{max}}}{\Lambda_{1}}}\log{\lr{\frac{1}{\delta}}}\left(\log{(N)}+b\right)\right)$.
		\item Qubits: $\mathcal{O}(\log{(N)}+b)$.
	\end{itemize}
\end{restatable}
\begin{proof}
	Proof outline in~\cref{sec:USA}.
\end{proof}

Our algorithm for simulating sparse Hamiltonians splits  $H=\sum^m_{j=1}H_j$  into $m$ Hermitian terms, where matrix entries of the $j^\text{th}$ term are
\begin{align}
(H_j)_{ik}=
\begin{cases}
H_{ik}, & \Lambda^{(j-1)}_{\mathrm{max}} < |H_{ik}| \le \Lambda^{(j)}_{\mathrm{max}}, \\
0, &\text{otherwise},
\end{cases}
\quad
0=\Lambda^{(0)}_{\mathrm{max}}<\Lambda^{(1)}_{\mathrm{max}}<\cdots<\Lambda^{(m)}_{\mathrm{max}}=\Lambda_\mathrm{max}.
\end{align}
Each term is block-encoded with normalization constant $\alpha_j$ by a procedure making $C_j$ queries to ${O}_{H}$ and ${O}_{F}$. We may then simulate time-evolution $e^{-iHt}$ by recombining these terms using~\cref{Thm:HamSimRecursion}. The query complexity of this procedure then depends strongly on the scaling of $\alpha_jC_j$. A judicious choice of cut-offs $\Lambda^{(j)}_{\mathrm{max}}$ combined with the block-encoding procedure~\cref{thm:sparse-Ham-block-encoding} allows us to prove our main result of~\cref{thm:HamSim_sparse} on sparse Hamiltonian simulation.
\begin{proof}[Proof of~\cref{thm:HamSim_sparse}.]
	The Hamiltonian $H_j$ in the decomposition $H=\sum^m_{j=1}H_j$ has matrix elements $\Lambda^{(j-1)}_{\mathrm{max}} < |H_{ik}| \le \Lambda^{(j)}_{\mathrm{max}}$, and may be block encoded using the procedure of~\cref{thm:sparse-Ham-block-encoding}. The only difference is replacing the oracle $O_H$ with an oracle 
	\begin{align}
	O_{H_j}\ket{i}\ket{k}\ket{z}=\ket{i}\ket{k}\ket{z\oplus ({H}_j)_{ik}},
	\end{align}
	that outputs matrix entries of $H_j$. Using $\mathcal{O}(1)$ queries to the oracle $O_H$ and $\mathcal{O}(\op{poly}(b))$ quantum gates, $O_{H_j}$ is constructed by computing the absolute value of $|({H})_{ik}|$ on $\mathcal{O}(b)$ bits, and performing a comparison with $\Lambda^{(j-1)}_{\mathrm{max}}$ and $\Lambda^{(j)}_{\mathrm{max}}$.

From~\cref{thm:sparse-Ham-block-encoding}, we may block-encode a Hamiltonian $\tilde{H}_j$ that approximates $\|\tilde{H}_j - H_j\| = \mathcal{O}(\Lambda^{(j)} \delta_j)$, with a normalization constant $\alpha_j=\Theta(\Lambda^{(j)}_1)\ge \|H_j\|_1$.
Thus the error from simulating time-evolution by $\tilde{H}=\sum_{j=1}^m \tilde{H}_j$ instead of $H$ is bounded by
\begin{align}
\|e^{-i\tilde{H}t}-e^{-iHt}\|
\le t\|\tilde{H}-H\|\le t\sum_{j=1}^m\|\tilde{H}_j-H_j\|
=\mathcal{O}\left(t\sum_{j=1}^m\Lambda^{(j)}\delta_j\right).
\end{align}
Using the upper bound $\Lambda^{(j)}\le \Lambda^{(j)}_1\le \sqrt{d}\Lambda^{(j)}_{1 \shortto 2} \le \sqrt{d}\Lambda_{1 \shortto 2}$, the overall contribution of this error may be bounded by $\mathcal{O}(\epsilon)$ with the choice $\delta_{j}=\frac{\epsilon}{m t \sqrt{d}\Lambda_{1 \shortto 2}}$. Thus the query complexity of block-encoding $H_j$ is $C_j=\mathcal{O}(\sqrt{d\Lambda^{(j)}_\mathrm{max}/\Lambda^{(j)}_1}\log{(mt\sqrt{d}\Lambda_{1 \shortto 2}/\epsilon)})$. Using these values $C_j$, and relabeling the $\alpha_j$ so that they are sorted like $\alpha_1 \ge \alpha_2 \ge\cdots \ge \alpha_m$, the query complexity of simulation by~\cref{Thm:HamSimRecursion} is
\begin{align}
\label{eq:sparse_cost_intermediate}
C_\mathrm{queries}[H,t,\epsilon]&=
\mathcal{O}\left(t\langle\vec{\alpha},\vec{C}\rangle\log^{2m-1}{\lr{\frac{t\alpha_1}{\epsilon}}} \right)
\\\nonumber
&=\mathcal{O}\left(t\left(
\sum^m_{j=1}\sqrt{d\Lambda^{(j)}_\mathrm{max}\Lambda^{(j)}_1}\right)(\log(t \sqrt{d}\Lambda_{1 \shortto 2}/\epsilon))^{\mathcal{O}(m)} \right).
\end{align}
In the last line, we simplify $\log{\lr{\frac{mt\sqrt{d}\Lambda_{1 \shortto 2}}{\epsilon}}}\log^{2m-1}{\lr{\frac{t\sqrt{d}\Lambda_{1 \shortto 2}}{\epsilon}}}=\log^{\mathcal{O}(m)}{\lr{\frac{t\sqrt{d}\Lambda_{1 \shortto 2}}{\epsilon}}}$.

We use the following upper bounds on the induced one-norm of $H_j$:
\begin{align}
\text{For } j=1,\; &\|H_1\|_1 
\le \|H\|_1
= \max_j \|H\cdot e_{j}\|_1 
\le \sqrt{d} \max_j\|H\cdot e_{j}\|
=\sqrt{d}\|H\|_{1 \shortto 2}
\le \sqrt{d}\Lambda_{1 \shortto 2} =\Lambda^{(1)}_1,
\\\nonumber
\forall j>1,\; &\|H_j\|_1
=\max_{k}\sum_{i}|(H_j)_{ik}|
< \max_{k}\sum_{i}\frac{|H_{ik}|^2}{\Lambda^{(j-1)}_{\mathrm{max}}}
= \frac{\|H\|^2_{1 \shortto 2}}{\Lambda^{(j-1)}_{\mathrm{max}}}
\le \frac{\Lambda^2_{1 \shortto 2}}{\Lambda^{(j-1)}_{\mathrm{max}}}
=\Lambda^{(j)}_1.
\end{align}
In the first line, $e_j$ is a unit vector with an entry $1$ in row $j$, and we apply the fact that $\|v\|_1\le\sqrt{d}\|v\|$ for all vectors $v$ with $d$ non-zero entries. In the second line, we apply the facts $|(H_j)_{ik}|/ \Lambda^{(j-1)}_{\mathrm{max}}\ge 1$ and $|H_{ik}|\ge|(H_j)_{ik}|$. By substituting into~\cref{eq:sparse_cost_intermediate}, we obtain
\begin{align}
\label{eq:sparse_cost_intermediateB}
C_\mathrm{queries}[H,t,\epsilon]
&=\mathcal{O}\left(t\left[
\sqrt{d^{3/2}\Lambda^{(1)}_\mathrm{max}\Lambda_{1 \shortto 2}}
+\sum^m_{j=2}\sqrt{d\frac{\Lambda^{(j)}_\mathrm{max}}{\Lambda^{(j-1)}_\mathrm{max}}}\Lambda_{1 \shortto 2}\right](\log(t \sqrt{d}\Lambda_{1 \shortto 2}/\epsilon))^{\mathcal{O}(m)} \right).
\end{align}

We now present an appropriate sequence of cut-offs $\Lambda^{(j)}_{\mathrm{max}}$, chosen so that all terms in the square brackets of~\cref{eq:sparse_cost_intermediateB} scale identically.
The largest may be chosen to be $\Lambda^{(m)}_{\mathrm{max}}=\Lambda_{1 \shortto 2}$, which follows from the inequality $\|H\|_{1 \shortto 2}=\max_{k}\sqrt{\sum_{i}|H_{ik}|^2}\ge \|H\|_\mathrm{max}$. The smallest may be chosen to be $\Lambda^{(1)}_{\mathrm{max}}=\Lambda_{1 \shortto 2}d^{-1/2+\gamma}$, for some $\gamma>0$. For any fixed value of $m$,  let us interpolate between these extremes with a fixed ratio
\begin{align}
\label{eq:ratio_choice}
\frac{\Lambda^{(j)}_{\mathrm{max}}}{\Lambda^{(j-1)}_{\mathrm{max}}}
=d^{\gamma} >1
\quad
\Rightarrow 
\quad 
\gamma=\frac{1}{2m}
\quad\Rightarrow
\quad
\Lambda^{(j)}_{\mathrm{max}}=\Lambda_{1 \shortto 2}d^{\frac{1}{2}\frac{j}{m}-\frac{1}{2}}
\quad\Rightarrow
\quad
\Lambda^{(j)}_1 = \sqrt{d} \Lambda_{1 \shortto 2}d^{-\frac{j-1}{2m}}.
\end{align}
Substituting this choice into~\cref{eq:sparse_cost_intermediateB}, the cost of simulation is
\begin{align}
\label{eq:sparse_cost_intermediate1}
C_\mathrm{queries}[H,t,\epsilon]&=\mathcal{O}\left(t
\sqrt{d}\Lambda_{1 \shortto 2}d^{\frac{1}{4m}}(\log(t\sqrt{d}\Lambda_{1 \shortto 2}/\epsilon))^{\mathcal{O}(m)} \right)
\\\nonumber
&=\mathcal{O}\left( t
\sqrt{d}\Lambda_{1 \shortto 2}e^{\mathcal{O}(m \log\log(t\sqrt{d}\Lambda_{1 \shortto 2}/\epsilon))+\frac{1}{4m}\log{(d)}}\right)
\\\nonumber
&=\mathcal{O}\left( t
\sqrt{d}\Lambda_{1 \shortto 2}e^{\mathcal{O}(\sqrt{\log{(d)}}\log\log(t\sqrt{d}\Lambda_{1 \shortto 2}/\epsilon))}\right)
\\\nonumber
&=\mathcal{O}\left(t
\sqrt{d}\Lambda_{1 \shortto 2}\left(\frac{t\sqrt{d}\Lambda_{1 \shortto 2}}{\epsilon}\right)^{o(1)}\right).
\end{align}
In the first line, we simplify using $m\log^{\mathcal{O}(m)}{\lr{\frac{t\sqrt{d}\Lambda_{1 \shortto 2}}{\epsilon}}}=\log^{\mathcal{O}(m)}{\lr{\frac{t\sqrt{d}\Lambda_{1 \shortto 2}}{\epsilon}}}$.
In the third line, we minimize cost with respect to $d$ by choosing 
$m=\mathcal{O}\left(\sqrt{\log{(d)}}\right)$. 
In the last line, the query complexity in~\cref{thm:HamSim_sparse} is proven by noting that the factor $e^{\mathcal{O}({\sqrt{\log{(d)}}\log{\log( t\sqrt{d}\Lambda_{1 \shortto 2}/\epsilon)}})}=\mathcal{O}\left(\left(\frac{t\sqrt{d}\Lambda_{1 \shortto 2}}{\epsilon}\right)^{o(1)}\right)$ has subpolynomial scaling with respect to $t,d,\Lambda_{1 \shortto 2}$ and $\epsilon$. The gate complexity is proven simply by multiplying the query complexity by the block-encoding overhead of $\mathcal{O}(\log{(N)}+b^k)$ gates per query.
\end{proof}

Simple modifications allow for minor improvements in query complexity. For instance, we used an upper bound $\Lambda_\mathrm{max}=\Lambda_{1 \shortto 2}$ in~\cref{eq:ratio_choice}. This enabled expressing complexity in terms of just the matrix norm $\Lambda_{1 \shortto 2}$. However, repeating our proof with both $\Lambda_\mathrm{max}$ and $\Lambda_{1 \shortto 2}$ free parameters leads to $C_\mathrm{queries}=\mathcal{O}\left( t
\sqrt{d}\Lambda_{1 \shortto 2}e^{\mathcal{O}(\sqrt{\log{(\sqrt{d}\Lambda_{\max}/\Lambda_{1 \shortto 2})}}\log\log(t\sqrt{d}\Lambda_{1 \shortto 2}/\epsilon))}\right)$. This implies polylogarithmic, instead of subpolynomial, scaling with error in the special case $\sqrt{d}\Lambda_{\max}=\mathcal{O}(\Lambda_{1 \shortto 2})$.

\section{Sparse Hamiltonian simulation lower bound}
\label{Sec:Lower_Bound}
We now prove a lower bound demonstrating that the scaling of our simulation algorithm with $t\sqrt{d}\|H\|_{1 \shortto 2}$ is optimal up to sub-polynomial factors. The argument is identical to the lower bound by~\cite{Low2017USA}. The only difference is that we quantify cost with the subordinate norm $\|H\|_{1 \shortto 2}$ instead of the induced one-norm $\|H\|_1$. This leads to a lower bound of $\Omega\left(t \sqrt{d}\|H\|_{1 \shortto 2}\right)$, which is a more general result. 
\HamSimLowerBound*
\begin{proof}
	We construct sparse Hamiltonians whose dynamics compute suitable Boolean functions with known quantum query lower bounds. The first Hamiltonian ${H}_{\op{PARITY}}$ computes the parity of $n$ bits, and the second Hamiltonian ${H}_{\op{OR}}$ computes the disjunction of $m$ bits. By composing these Hamiltonians, one may obtain a third Hamiltonian ${H}_{\op{PARITY}\circ\op{OR}}$ that computes the parity of $n$ bits, where each bit is the disjunction of $m$ bits, as depicted in~\cref{Fig:ParityOr}. The stated lower bound is then obtained by combining the known quantum query complexity of $\Omega(n\sqrt{m})$~\cite{Reichardt2009span} for computing $\op{PARITY}\circ\op{OR}$ on $n\times m$ bits~, with parameters that describe ${H}_{\op{PARITY}\circ\op{OR}}$.
\begin{figure} [t]
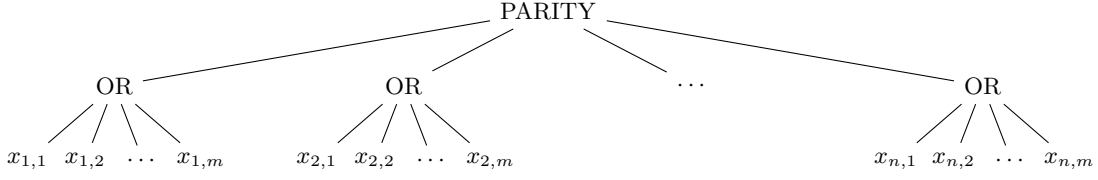

	\centering
	\tikz [font=\footnotesize,
	level 1/.style={sibling distance=10em},
	level 2/.style={sibling distance=2em}, level distance=1cm]
	\node {$\op{PARITY}$} 
	child { node {$\op{OR}$}
		child { node {$x_{1,1}$} }
		child { node {$x_{1,2}$} }
		child { node {$\cdots$} }
		child { node {$x_{1,m}$} }
	}
	child { node {$\op{OR}$}
		child { node {$x_{2,1}$} }
		child { node {$x_{2,2}$} }
		child { node {$\cdots$} }
		child { node {$x_{2,m}$} }
	}
	child { node {$\cdots$}
	}
	child { node {$\op{OR}$}
		child { node {$x_{n,1}$} }
		child { node {$x_{n,2}$} }
		child { node {$\cdots$} }
		child { node {$x_{n,m}$} }
	};
	\caption{\label{Fig:ParityOr}Computation of $\op{PARITY}\circ\op{OR}$ on $n\times m$ bits $x_{i,j}\in\{0,1\}$.}
\end{figure}	

The Hamiltonian ${H}_{\op{PARITY}}$ by~\cite{Berry2015Hamiltonian} is constructed from a simpler $(n+1)\times(n+1)$ Hamiltonian ${H}_{\mathrm{spin}}$ that acts on basis states $\{\ket{j}_s : j\in\{0,\cdots,n\}\}$ with non-zero matrix elements $\bra{j-1}_s{H}_{\mathrm{spin}}\ket{j}_s
=
\bra{j}_s{H}_{\mathrm{spin}}\ket{j-1}_s
=\sqrt{j(n-j+1)}/n$ for $j\in\{1,\cdots, N\}$. Thus
\begin{align}
{H}_{\mathrm{spin}}=\sum_{j\in\{1,\cdots N\}}\frac{\sqrt{j(n-j+1)}}{n}\ket{j-1}\bra{j}_{s} + \op{h.c.}.
\end{align}
The transitions generated by Hamiltonian may be represented by the graph in~\cref{Fig:H_parity_graph}a, where basis states are nodes, and non-matrix elements are edges between nodes.
with the useful property that time-evolution by ${H}_{\mathrm{spin}}$ for time $\frac{n\pi}{2}$ transfers the state $\ket{0}_s$ to $\ket{n}_s=e^{-i{H}_{\mathrm{spin}}n\pi/2}\ket{0}_s$, after passing through all intermediate nodes.
\begin{figure} [t]
	\centering
\begin{tikzpicture}
\node [set=label] (m0) at (-2,4) {a)};
\foreach \i in {0,...,9} {
	\node [set=label] (l\i) at (\i,3.7) {$\ket{\i}_s$};
	\node [circle, draw=black, fill=black, inner sep=2pt] (x\i) at (\i,3) {};
}
\graph{
	(x0) -- (x9);
};

\node [set=label] (m0) at (-2,2) {b)};
\foreach \i in {0,...,9} {
	\node [set=label] (l\i) at (\i,1.7) {$\ket{\i}_s$};
	\node [circle, draw=black, fill=black, inner sep=2pt] (x\i) at (\i,1) {};
	\node [circle, draw=black, fill=black, inner sep=2pt] (y\i) at (\i,0) {};
	\node [circle, draw=black, fill=black, inner sep=2pt] (z\i) at (\i,0) {};
}
\node [set=label] () at (-1,1) {$\ket{0}_{\mathrm{out}}$};
\node [set=label] () at (-1,0) {$\ket{1}_{\mathrm{out}}$};
\graph{
	(y0) -- (y1) -- (x2) -- (x3) -- (y4) -- (y5) -- (y6) -- (x7) -- (y8) -- (y9);
	(x0) -- (x1) -- (y2) -- (y3) -- (x4) -- (x5) -- (x6) -- (y7) -- (x8) -- (x9);
};
\node [set=label] (m0) at (-1,-0.6) {$x_j=$};
\foreach \i/\val in {0/0,1/1,2/0,3/1,4/0,5/0,6/1,7/1,8/0} {
	\node [set=label] (m0) at (\i+0.5,-0.5) {$\val$};
}
\end{tikzpicture}

\caption{\label{Fig:H_parity_graph}a) Graph representation of non-zero matrix elements of the Hamiltonian ${H}_{\mathrm{spin}}$ with $n=9$. Evolution under ${H}_{\mathrm{spin}}$ for time $n\pi/2$ transfers state $\ket{0}_s$ to $\ket{9}_s$. b) Graph representation of non-zero matrix elements of the Hamiltonian ${H}_{\mathrm{PARITY}}$ with $n=9$, obtained by composing ${H}_{\mathrm{spin}}$ element-wise with ${H}_{\mathrm{NOT},j}$, which depends on the $j^{\text{th}}$ value of the bit-sting $x\in\{0,1\}^n$. Evolution under ${H}_{\mathrm{PARITY}}$ for time $n\pi/2$ transfers state $\ket{0}_s\ket{z}_\mathrm{out}$ to $\ket{9}_s\ket{z\oplus_j x_j}_\mathrm{out}$.}
\end{figure}
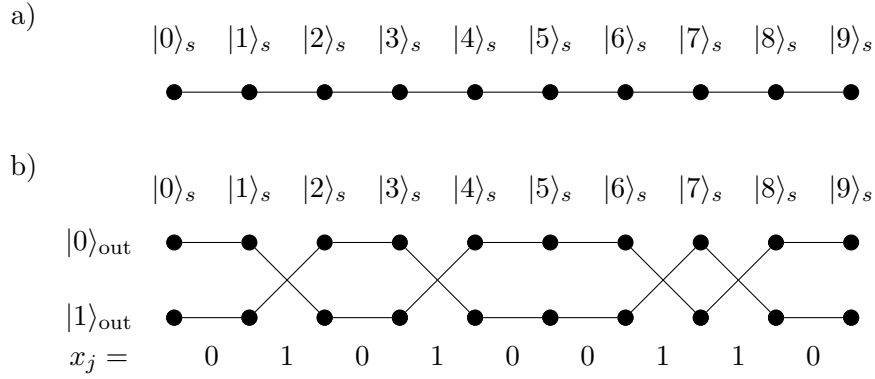

We now modify ${H}_{\mathrm{spin}}$ based on a bit-string $x\in\{0,1\}^{n}$. For each bit $x_j$, consider the $2\times 2$ Hamiltonian $H_{\op{NOT},j}$ that acts on basis states $\{\ket{k}_{\mathrm{out}}:k\in\{0,1\}$, with matrix elements defined by
\begin{align}
{H}_{\mathrm{NOT},j}=\left(
\begin{matrix}
x_j \oplus 1 & x_j \\
x_j & x_j \oplus 1
\end{matrix}
\right).
\end{align}
Observe that ${H}_{\mathrm{NOT},j}\ket{k}_\mathrm{out}=\ket{k\oplus x_j}_\mathrm{out}$. Composing ${H}_{\op{spin}}$ with ${H}_{\mathrm{NOT},j}$ defines the $(2n+2)\times (2n+2)$ ${H}_{\mathrm{PARITY}}$ in the following manner
\begin{align}
\label{Eq:Ham_Parity}
{H}_{\mathrm{PARITY}}=
\sum_{j\in\{1,\cdots N\}}\frac{\sqrt{j(n-j+1)}}{n}\ket{j-1}\bra{j}_{s}\otimes {H}_{\mathrm{NOT},j}  + \op{h.c.}.
\end{align}
This Hamiltonian is represented by the graph in~\cref{Fig:H_parity_graph}b. Note the two disjoint paths connecting states $\ket{0}_s$ and $\ket{n}_s$. As the path taken by an initial state $\ket{0}_s\ket{0}_\mathrm{out}$ depends on the parity of $x$, time-evolution by ${H}_{\mathrm{PARITY}}$ for time $\frac{n\pi}{2}$ transfers the state $\ket{0}_s\ket{0}_\mathrm{out}$ to $\ket{n}_s\ket{\bigoplus_j x_j}_\mathrm{out}$. This computes the parity of $x$, and the answer is obtained by measuring the `$\op{out}$' register.

The $2m\times 2m$ Hamiltonian ${H}_{\op{OR}}$ by~\cite{Low2017USA} computes the $\op{OR}$ of a bit-string $x\in\{0,1\}^m$, assuming that at most one bit is non-zero. In the basis $\{\ket{k}_{\mathrm{out}}\ket{l}_{o}:k\in\{0,1\},l\in\{1,\cdots,m\}\}$, its matrix elements are defined by
\begin{align}
\label{eq:Ham_OR}
{H}_{\mathrm{OR}}&=\left(
\begin{array}{c|c}
{C}_{ 1} & {C}_{0} \\
\hline
{C}^\dag_{0} & {C}_{1}
\end{array}
\right),\\\nonumber
{C}_{0}&=\left(
\begin{array}{cccc}
x_1 & x_2 & \cdots  &x_{m} \\
x_{m} & x_{1} & \cdots & x_{m-1}\\
x_{m-1} & x_{m} & \cdots  & x_{m-2}\\
\vdots & \vdots & \ddots & \vdots \\
x_{2} & x_{3} & \cdots & x_{1} 
\end{array}
\right), 
\quad
{C}_{1} = \frac{1}{m}\left(
\begin{array}{cccc}
1 & 1 & \cdots  &1 \\
1 & 1 & \cdots  &1 \\
\vdots & \vdots & \ddots  &\vdots \\
1 & 1 & \cdots  &1 
\end{array}
\right)
-\frac{{C}_{0}+{C}_{0}^\dag}{2}.
\end{align}
It is easy to verify that if at most one bit in $x$ is non-zero, ${H}_{\mathrm{OR}}\ket{k}_\mathrm{out}\ket{u}_o=\ket{k\oplus \mathrm{OR}(x)}_\mathrm{out}\ket{u}_o$, where $\ket{u}_o=\frac{1}{\sqrt{m}}\sum_{l\in\{1,\cdots,m\}}\ket{l}_o$ is a uniform superposition state. 

As the function $\op{PARITY}\circ\op{OR}$ we consider acts on $n\times m$ bits, let ${H}_{\op{OR},_j}$ be the Hamiltonians ${H}_{\mathrm{OR}}$ of~\cref{eq:Ham_OR}, except that the input bit-string $x$ is replaced by the $j^{\text{th}}$ set of $m$ bits $x_j=(x_{j,1},\cdots,x_{j,m})\in\{0,1\}^m$. By replacing each ${H}_{\mathrm{NOT},j}$ with ${H}_{\op{OR},_j}$, we obtain the $m(n+1)\times m(n+1)$ Hamiltonian
\begin{align}
\label{Eq:Ham_Parity_OR}
{H}_{\mathrm{PARITY}\circ \mathrm{OR}}=
\sum_{j\in\{1,\cdots N\}}\frac{\sqrt{j(n-j+1)}}{n}\ket{j-1}\bra{j}_{s}\otimes {H}_{\mathrm{OR},j}  + \op{h.c.}.
\end{align}
Time-evolution by ${H}_{\mathrm{PARITY}\circ \mathrm{OR}}$ for time $\frac{n\pi}{2}$ transforms the state
\begin{align}
e^{-i{H}_{\mathrm{PARITY}\circ \mathrm{OR}}n\pi/2} \ket{0}_s\ket{u}_o\ket{0}_\mathrm{out} = \ket{n}_s\ket{u}_o\ket{\oplus_j \mathrm{OR}(x_j)}_\mathrm{out},
\end{align}
thus measuring the `$\op{out}$' register returns the parity of $n$ bits, where each bit is the disjunction of $m$ bits.

As the desired lower bound requires us to independently vary over three parameters, we introduce one final modification. Let ${H}_{\mathrm{complete}}$ be a $s\times s$ Hamiltonian where all matrix elements are $1$ in the basis $\{\ket{i}_c:i\in\{1,\cdots,s\}\}$. We now take the tensor product of ${H}_{\mathrm{PARITY}\circ \mathrm{OR}}$ with ${H}_{\mathrm{complete}}$, the resulting $sm(n+1)\times sm(n+1)$ Hamiltonian
\begin{align}
{H}={H}_{\mathrm{PARITY}\circ \mathrm{OR}}\otimes {H}_{\mathrm{complete}}.
\end{align}
As the uniform superposition $\ket{u}_c=\frac{1}{\sqrt{s}}\sum_{i\in\{1,\cdots,s\}}\ket{i}_c$ state is an eigenstate of ${H}_{\mathrm{complete}}\ket{u}_c=s\ket{u}_c$ with eigenvalue $s$, time-evolution by ${H}$ with initial state $\ket{0}_s\ket{u}_o\ket{u}_c\ket{0}_\mathrm{out}$ performs the same computation as~\cref{Eq:Ham_Parity_OR}, but in shorter time $\frac{n\pi}{2s}$.

By varying the problem size through the number of bits $m,n$, and the dimension $s$, we may express the query lower bound $\Omega(n\sqrt{m})$ in terms of the evolution time $t$, and sparsity $d$ of $H$, and the subordinate norm $\|H\|_{1 \shortto 2}$. We note the following facts:
\begin{itemize}
	\item The max-norm $\|H\|_\mathrm{max}=\mathcal{O}(1)$.
	\item The evolution time $t=\Theta\left(\frac{n}{s}\right)$.
	\item The sparsity of $H$ is $d=\Theta\left(ms\right)$.
	\item The subordinate norm of $H$ is $\|H\|_{1 \shortto 2}=\max_{j}\sqrt{\sum_{i}|H_{ij}|^2}
	=\mathcal{O}\left(\sqrt{s}\right)
	$.
\end{itemize}
Substituting these parameters into the lower bound, we obtain the stated quantum query complexity for sparse Hamiltonian simulation of
\begin{align}
\Omega(n\sqrt{m})=\Omega(ts\sqrt{m})=\Omega(t\sqrt{d}\sqrt{s})=\Omega(t\sqrt{d}\|H\|_{1 \shortto 2}).
\end{align}
\end{proof}


\section{Application to black-box unitaries}
\label{sec:blackboxunitary}
Following~\cite{Jordan2009,Berry2012}, the black-box unitary problem reduces to an instance of Hamiltonian simulation, which we implement using the algorithm of~\cref{thm:HamSim_sparse}. The reduction is straightforward. 
\begin{proof}[Proof of~\cref{thm:BlackBoxUnitary}.]
For any $N\times N$ unitary $U$ operator that acts on basis states $\{\ket{j}_u:j\in\{0,\cdots,N-1\}\}$, let us define a $2N\times 2N$ Hamiltonian $H$ that acts on basis states $\{\ket{k}_h\ket{j}_u:k\in\{0,1\}:j\in\{0,\cdots,N-1\}\}$:
\begin{align}
\label{eq:BlackBoxUnitaryHamiltonian}
{H}&=\left(
\begin{array}{cc}
0 & U \\
U^\dag & 0
\end{array}
\right).
\end{align}
Using the fact $H^2 = \ii$, the time-evolution operator generated by $H$ has a simple form:
\begin{align}
e^{-iHt}=\cos{(t)}-i\sin{(t)}H
\quad\Rightarrow\quad
e^{-iH\pi/2}=-iH.
\end{align}
Thus we may apply $U$ to an arbitrary state $\ket{\psi}_u$, up to a global phase, by applying 
\begin{align}
e^{-iH\pi/2}\ket{1}_h\ket{\psi}_u=-i \ket{0}_hU\ket{\psi}_u.
\end{align}
The $-i\ket{0}_h$ state may be converted to $\ket{1}_h$ by a single Pauli $Y$ gate. Thus the cost of applying $U$ reduces to the cost of simulating the Hamiltonian $H$ for constant time $t=\pi/2$. 

	Given that matrix positions and values of a $d$-sparse $U$ are described by the sparse matrix oracles of~\cref{def:Sparse_Oracle}, $\mathcal{O}(1)$ queries suffice to synthesize black-box oracles that describe the $d$-sparse Hamiltonian $H$ of~\cref{eq:BlackBoxUnitaryHamiltonian}. As $U$ is unitary, this Hamiltonian has max-norm $\|H\|_\mathrm{max}=\|U\|_\mathrm{max}\le\Lambda_\mathrm{max}=1$, and subordinate norm $\|H\|_{1 \shortto 2}=\|U\|_{1 \shortto 2}\le\|U\|\le \Lambda_{1 \shortto 2}=1$. By substituting these parameters into~\cref{thm:HamSim_sparse}, we immediately obtain the query complexity of 
	\begin{align}
	\mathcal{O}(\sqrt{d}(d/\epsilon)^{o(1)}).
	\end{align}
	for approximating black-box unitaries.
\end{proof}

\section{Application to sparse systems of linear equations}
\label{sec:blackboxQLSP}
Following~\cite{Childs2015LinearSystems,Chakraborty2018BlockEncoding}, the cost of solving systems of linear equations depends primarily on the cost of Hamiltonian simulation, which we once again implement using the algorithm of~\cref{thm:HamSim_sparse}. As details of the reduction are quite involved, we only sketch the proof, and obtain~\cref{thm:BlackBoxQLSP} by invoking results in prior art. 

A system of linear equations is described by a matrix $A\in\mathbb{C}^{N\times N}$ typically characterized by spectral norm $\|A\|=1$ and condition number $\|A^{-1}\|\le\kappa$, and an input vector $\ket{b}\in\mathbb{C}^{N}$. This is solved by preparing a state proportional to $A^{-1}\ket{b}$. Without loss of generality, we may assume that $A$ is Hermitian~\cite{Harrow2009}. Further assuming that $A/\alpha$ is block-encoded by a unitary $U_A$, as in~\cref{Def:Standard_Form}, one may use linear-combination-of-unitaries~\cite{Childs2015LinearSystems} or quantum signal processing~\cite{Haah2018product} to block-encode $A^{-1}/\kappa$ in a unitary $U_{A^{-1}}$. Generally, $U_{A^{-1}}$ may be approximated to error $\epsilon$, meaning that $\|A^{-1}-\kappa (\bra{0}\otimes I)U_{A^{-1}}(\ket{0}\otimes I)\|\le\epsilon$, using $\mathcal{O}(\alpha \kappa \plog{(\alpha\kappa/\epsilon)})$ queries to $U_A$. 

In the basic approach, applying $U_{A^{-1}}\ket{0}\ket{b}\approx\ket{0}\frac{A^{-1}}{\kappa}\ket{b}+\cdots$ approximates the desired state, but with a worst-case success probability of $\mathcal{O}(\kappa^{-2})$. Thus $\mathcal{O}(\kappa)$ rounds of amplitude amplification are required to obtain the desired state with $\mathcal{O}(1)$ success probability. By multiplying these factors, this leads to an overall query complexity of $\mathcal{O}(\alpha \kappa^2 \plog{(\alpha\kappa/\epsilon)})$ to $U_A$ and the unitary oracle $O_b\ket{0}=\ket{b}$ that prepares the input state $\ket{b}$. However, this may be improved to $\mathcal{O}(\alpha \kappa \plog{(\alpha\kappa/\epsilon)})$ using a more sophisticated approach based on variable-time amplitude amplification as follows.
\begin{restatable}[Variable-time quantum linear systems algorithm; adapted from Lemma 27 of~\cite{Chakraborty2018BlockEncoding}]{lemma}{VTAQLSP}
	\label{thm:QLSP_Chakraborty}
	Let $\ket{b}\in\mathbb{C}^N$, $\kappa=\Omega(1)$, and $A\in\mathbb{C}^{N\times N}$ be a Hermitian matrix such that $\|A\|=1$ and $\|A^{-1}\|\le\kappa$. Suppose that $A/\alpha$ is block-encoded by a unitary $U_A$ with error $o(\epsilon/\operatorname{poly}(\kappa,\log{(1/\epsilon)}))$, and that $\ket{b}$ is prepared by a unitary oracle $O_b\ket{0}=\ket{b}$. Then there exists a quantum algorithm that outputs a state $\ket{\psi}$ such that $\left\|\ket{\psi}-\frac{A^{-1}\ket{b}}{\|A^{-1}\ket{b}\|}\right\|\le\epsilon$. The query complexity of this algorithm to $U_A$ and $O_b$ is
	\begin{align}
	 \mathcal{O}(\alpha \kappa \plog{(\alpha\kappa/\epsilon)}).
	\end{align}
\end{restatable}

\begin{proof}[Proof of~\cref{thm:BlackBoxQLSP}]
Now assuming that $A$ is $d$-sparse and described by the black-box oracles of~\cref{def:Sparse_Oracle}, our simulation algorithm~\cref{thm:HamSim_sparse} allows us to approximate the time-evolution operator $e^{-iA/2}$ with error $\delta$ using $\mathcal{O}(\sqrt{d}(d/\delta)^{o(1)})$ queries. Note that we have used the fact that $\|A\|=1$ to bound the subordinate norm input $\Lambda_{1 \shortto 2}=1$ required by the algorithm. By taking a matrix logarithm of this operator using Theorem 9 of~\cite{Low2017USA}, $A/\alpha$ may be block-encoded with $\alpha=\Theta(1)$ in a unitary $U_A$ without affecting the query complexity. Thus by invoking~\cref{thm:QLSP_Chakraborty} which requires $\delta=o(\epsilon/\operatorname{poly}(\kappa,\log{(1/\epsilon)}))$, an $\epsilon$-approximation of the state $\frac{A^{-1}\ket{b}}{\|A^{-1}\ket{b}}$ may be prepared. The query complexity is obtained by multiplication, and is
	\begin{align}
	\mathcal{O}(\sqrt{d}(d/\delta)^{o(1)})\times\mathcal{O}(\kappa \plog{(\kappa/\epsilon)})=\mathcal{O}(\kappa\sqrt{d}(\kappa d/\epsilon)^{o(1)}).
	\end{align}
\end{proof}

\section{Conclusion}
\label{sec:conclusion}
Our algorithm for sparse Hamiltonian simulation combines ideas from simulation in the interaction picture with uniform spectral amplification. In applications such solving linear systems of equations or implementing black-box unitaries, one often simulates a sparse Hamiltonian where parameters for the time $t$, sparsity $d$, and subordinate norm $\|H\|_{1 \shortto 2}$, naturally describe the problem. Given these, our algorithm scales like $\mathcal{O}\left(( t
\sqrt{d}\|H\|_{1 \shortto 2}(
t\sqrt{d}\|H\|_{1 \shortto 2}/\epsilon)^{o(1)}\right)$, which is optimal up to subpolynomial factors. Moreover, one is allowed to substitute these parameters for any weaker set of constraints, such as from a well-known sequence of tight norm inequalities for sparse Hamiltonians~\cite{Childs2010Limitation},
\begin{align}
\nonumber
\|H\|_\mathrm{max} \le \|H\|_{1 \shortto 2}  \le \|H\| \le \|\op{abs}(H)\|\le\|H\|_1\le\sqrt{d}\|H\|_{1 \shortto 2}\le\sqrt{d\|H\|_\mathrm{max}\|H\|_1}\le d\|H\|_\mathrm{max}.
\end{align}
Thus our algorithm generalizes, and in some cases strictly improves, prior art scaling with parameters $(d,\|H\|_\mathrm{max})$~\cite{Berry2015Hamiltonian,Low2016HamSim} or $(d,\|H\|_\mathrm{max},\|H\|_1)$~\cite{Low2017USA} or $(d,\|H\|)$~\cite{Berry2012}.

This greatly narrows the interval containing ultimate bounds on the complexity of sparse Hamiltonian simulation in the black-box setting. Known lower bounds forbid algorithms scaling like $\mathcal{O}(t\op{poly}(\|H\|))$ or even $\mathcal{O}(t\sqrt{d}\|H\|/\plog{(d)})$, so it would be interesting to pin this down, or improve the subpolynomial factors in our upper bound. Another useful direction would be to consider the simulation of structured Hamiltonians. For instance, some algorithms are highly successful at exploiting the geometric locality~\cite{Haah2018quantum} of certain Hamiltonians, or Hamiltonians with a large separation of energy scales~\cite{Low2018IntPicSim}. Within the black-box setting, the main challenge would to identify parameters that are sufficiently structured so as to enable a speedup, yet sufficiently general so as to describe problems of interest.

\textbf{Acknowledgements} -- We thank Dominic Berry, Andrew Childs, Robin Kothari, and Yuan Su for insightful discussions. In particular, we thank Nathan Wiebe for the idea of interaction picture simulation~\cite{Low2018IntPicSim}.


\appendix
\section{Hamiltonian simulation in the interaction picture}
\label{sec:IntPicSim}
Given a time-independent Hamiltonian $H=A+B$, the time-evolution of any quantum state is described by $\ket{\psi(t)}=e^{-iHt}\ket{\psi(0)}$. The time-evolution operator $e^{-iHt}$ is such that this state solves the Schr\"{o}dinger equation $i\partial_t\ket{\psi(t)}=H\ket{\psi(t)}$. An equivalent representation is the interaction picture, where we consider the time-evolution of a state $\ket{\psi_I(t)}=e^{iAt}\ket{\psi(t)}$. This state also evolves under the Schr\"{o}dinger equation, except that time-evolution is now generated by a time-dependent Hamiltonian $H_I(t)=e^{iAt}Be^{-iAt}$. This Hamiltonian is derived from the following sequence:
\begin{align}
\label{eq:time-dependent_schrodinger}
i\partial_t\ket{\psi_I(t)}=i\partial_t(e^{iAt}\ket{\psi(t)})=e^{iAt}(-A+i\partial_t)\ket{\psi(t)}=e^{iAt}B\ket{\psi(t)}=H_I(t)\ket{\psi_I(t)}.
\end{align}
The formal solution to~\cref{eq:time-dependent_schrodinger} is expressed by the time-ordered evolution operator $T(t_1,t_0)=\mathcal{T}\left[e^{-i\int_{t_0}^{t_1} H_I(s)\mathrm{d}s}\right]$, defined such that $\ket{\psi_I(t_1)}=T(t_1,t_0)\ket{\psi_I(t_0)}$ for any $t_1 \ge t_0$. Thus we may express the time-evolution operator as $e^{-iHt}=e^{-iAt}T(t,0)$. 

Our Hamiltonian simulation result~\cref{Thm:HamSimRecursion} is obtained by the recursive application of~\cref{thm:int_pic_sim}, which is a simulation algorithm by~\cite{Low2018IntPicSim}. This simulation algorithm consists of two steps: 1) the time-ordered evolution operator $T(t,0)$ is approximated with a time-dependent simulation algorithm; 2) $T(t,0)$ is combined with time-evolution by the $A$ component alone to obtain time-evolution by the full Hamiltonian $A+B$. This requires the following time-dependent simulation algorithm that was originally motivated by~\cite{Berry2015Truncated}, and then rigorously analyzed and optimized by~\cite{Low2018IntPicSim}.
\begin{lemma}[Time-dependent Hamiltonian simulation by a truncated Dyson series~\cite{Low2018IntPicSim}]
	\label{Thm:Compressed_TDS}
	Let $\tau>0$, and let $H(s): [0,\tau]\rightarrow \mathbb{C}^{N_s\times N_s}$ be a time-dependent Hamiltonian satisfying all the following:
	\begin{itemize}
		\item An upper bound on the spectral norm $\alpha \ge \max_{s}\|H(s)\|$ is known.
		\item An upper bound on the average time derivative $\alpha'\ge \frac{1}{\tau}\int^\tau_{0} \left\|\frac{\mathrm{d} H(s)}{ \mathrm{d} s}\right\| \mathrm{d}s$ is known.	
		\item There exists a unitary $U$ that block-encodes a block-diagonal Hamiltonian $H\in \mathbb{C}^{NM\times NM}$ with normalizing constant $\alpha$, where $M=\mathcal{O}\left( \frac{\tau^2}{\epsilon}\left(\frac{\alpha'}{\alpha}  +1\right)\right)$, and the $j^\text{th}$ block is $H(j\tau /M)$. Thus
		\begin{align}
		\label{eq:time_dep_block_encoding}
		(\bra{0}_a\otimes \ii_s)U(\ket{0}_a\otimes \ii_s) = \sum^{M-1}_{j=0}\ketbra{j}{j}_d\otimes\frac{H(j\tau/M)}{\alpha},
		\quad
		\ket{0}_a \in \mathbb{C}^{N_a},
		\quad
		\ket{j}_d \in \mathbb{C}^M.
		\end{align}
	\end{itemize}
	Then for all $\tau \le \frac{1}{2\alpha}$, the time-ordered evolution operator $\mathcal{T}\left[e^{-i\int_{t_0}^{t_1} H_I(s)\mathrm{d}s}\right]$ may be approximated to error $\epsilon$ using
	\begin{itemize}
		\item Queries to $U$: $\mathcal{O}\left(\frac{\log{(1/\epsilon)}}{\log\log{(1 /\epsilon)}}\right)$.
		\item Quantum gates: $\mathcal{O}\left((\log{(N_a)} + \log{(M)})\frac{\log{(1/\epsilon)}}{\log\log{(1/\epsilon)}}\right)$.
		\item Qubits: $\mathcal{O}(\log{(N_s)} + \log{(N_a)}+\log{(M)})$.
	\end{itemize}
\end{lemma}
Applying~\cref{Thm:Compressed_TDS} to the time-dependent Hamiltonian $H_I(t)$ leads to~\cref{thm:int_pic_sim}. For completeness, we restate this result and sketch the proof. 
\HamSimIntPic*
\begin{proof}
	Let us approximate the time-evolution operator $e^{-iH
		\tau}=e^{-iA\tau}T(\tau,0)$ for a short time-step $\tau=\mathcal{O}(\alpha_B^{-1})$. As $T(\tau,0)$ is generated by the time-dependent interaction-picture Hamiltonian $H_I(s)=e^{iAs}Be^{-iAs}$, we apply the time-dependent simulation algorithm~\cref{Thm:Compressed_TDS} to approximate $T(\tau,0)$ to error $\epsilon$. 
	
	The inputs required by~\cref{Thm:Compressed_TDS} are: 1) the spectral norm is $\max_s\|H_I(s)\|\le \alpha_B=\alpha$; 2) the time-derivative spectral norm is $\frac{1}{\tau}\int^\tau_{0} \left\|\frac{\mathrm{d} H_I(s)}{ \mathrm{d} s}\right\| \mathrm{d}s = \|[A,B]\|\le 2\alpha_A\alpha_B = \alpha'$; 3) there exists a block-encoding $U$ of $H_I(s)$, evaluated at times $s=j \tau/M$, where 
	$	M
	=\mathcal{O}\left( \frac{\tau^2}{\epsilon}\left(\frac{\alpha'}{\alpha}  +1\right)\right)
	=\mathcal{O}\left( \frac{\tau^2\alpha_A}{\epsilon}\right)
	$. 
	One possible implementation  of $U$ is
	\begin{align}
	U =(R^\dag\otimes I_A)\cdot(I_d \otimes U_B)\cdot(R\otimes I_A),
	\quad R=\left(\sum^{M-1}_{j=0}\ket{j}\bra{j}_d \otimes  e^{-iAj\tau /M}\right),
	\end{align}
	which uses one query to the block-encoding $U_B$ of $B$, and two applications each of controlled-$e^{iA2^k\tau/M}$ for positive integers $k\le\lceil\log_2{(M)}\rceil$. Thus the cost $C[U]$ of approximating $U$ to error $\epsilon$ is 
	\begin{align}
	C_\mathrm{queries}[U]&=\mathcal{O}(C_B + \log{(M)}C_\mathrm{queries}[A,\tau,\epsilon/\lceil\log_2{(M)}\rceil])=\mathcal{O}(C_B+\log{(M)}C_\mathrm{queries}[A,\tau,\epsilon]),
	\\\nonumber
	C_\mathrm{gates}[U]&=\mathcal{O}(\log{(M)}C_\mathrm{gates}[A,\tau,\epsilon/\lceil\log_2{(M)}\rceil])=\mathcal{O}(\log{(M)}C_\mathrm{gates}[A,\tau,\epsilon]),
	\end{align}
	where we have simplified the error dependence by applying the fact $\mathcal{O}(\log{(x\log{(x)})})=\mathcal{O}(\log{(x)})$ to
	$\log\left(\tau\alpha_A\lceil\log_2{(M)}\rceil/\epsilon\right)=\mathcal{O}\left(\log\left(M \log{(M)}\right)\right)=\mathcal{O}\left(\log\left(M\right)\right)=\mathcal{O}\lr{\log\left(\tau\alpha_A/\epsilon)\right)}$.

	Combining these inputs with the stated cost of time-dependent simulation in~\cref{Thm:Compressed_TDS} gives the cost $C[B,A,\tau,\epsilon]$ of approximating the time-evolution operator $e^{-iH\tau}$ to error $\epsilon=\epsilon_1+\frac{\log{(1/\epsilon_1)}}{\log\log{(1 /\epsilon_1)}}\epsilon_2+\epsilon_3$, where $\epsilon_1$ is the error contribution from the time-dependent simulation algorithm, $\epsilon_2$ is the error contribution from $R$, and $\epsilon_3$ is the error contribution from $e^{-iA\tau}$. The overall error is controlled by choosing $\epsilon_2=\mathcal{O}\lr{\epsilon_1/\frac{\log{(1/\epsilon_1)}}{\log\log{(1 /\epsilon_1)}}}$, $\epsilon_3=\mathcal{O}(\epsilon_1)$, and $\epsilon_1=\mathcal{O}(\epsilon)$ -- note that $C[A,t,\epsilon_2]=\mathcal{O}(C[A,t,\epsilon])$ as $\log{\lr{\frac{1}{\epsilon}\frac{\log{(1/\epsilon)}}{\log\log{(1 /\epsilon)}}}}=\mathcal{O}(\log{(1/\epsilon)})$. Thus the query complexity is	
	\begin{align}
	\label{eq:int_pic_sim_cost_tau}
	C_\mathrm{queries}[B,A,\tau,\epsilon]&=\mathcal{O}\left(\left(C_\mathrm{queries}[U]+C_\mathrm{queries}[A,\tau,\epsilon]\right)\frac{\log{(1/\epsilon)}}{\log\log{(1 /\epsilon)}}\right)
	\\\nonumber
	&=\mathcal{O}\left(\left(C_B+\log{(M)}C_\mathrm{queries}[A,\tau,\epsilon]\right)\frac{\log{(1/\epsilon)}}{\log\log{(1 /\epsilon)}}\right)
	\\\nonumber
	&=\mathcal{O}\left(\left(C_B+C_\mathrm{queries}[A,\tau,\epsilon] \log{\lr{\frac{\tau\alpha_A}{\epsilon}}}\right)\frac{\log{(1/\epsilon)}}{\log\log{(1 /\epsilon)}}\right),
	\end{align}
	and the gate complexity follows from a very similar derivation as follows.
	\begin{align}
	C_\mathrm{gates}[B,A,\tau,\epsilon]&=\mathcal{O}\left(\left(\log{(N_a)} + C_\mathrm{gates}[U]+C_\mathrm{gates}[A,\tau,\epsilon]+\log{(M)}\right)\frac{\log{(1/\epsilon)}}{\log\log{(1/\epsilon)}}\right)
	\\\nonumber
	&=\mathcal{O}\left((\tau\alpha_A+1)  \log^{\gamma+1}{\lr{\frac{\tau\alpha_A}{\epsilon}}}\log{(N_a)}\frac{\log{(1/\epsilon)}}{\log\log{(1/\epsilon)}}\right).
	\\\nonumber
	&=\mathcal{O}\left(C_\mathrm{gates}[A,\tau,\epsilon] \log{\lr{\frac{\tau\alpha_A}{\epsilon}}}\frac{\log{(1/\epsilon)}}{\log\log{(1/\epsilon)}}\right).
	\end{align}

	Evolution for longer times $t$ is achieved by $L=\mathcal{O}(t/\tau)$ applications of $e^{-iH	\tau}=e^{-iA\tau}T(\tau,0)$, each with error $\epsilon \rightarrow \epsilon\tau/t$. As~\cref{eq:int_pic_sim_cost_tau} holds for all $\tau=\mathcal{O}(\alpha_B^{-1})$, the simulation algorithm~\cref{Thm:Compressed_TDS} is always applied at least once. Thus the cost of time-evolution for all $t\alpha_B>0$ has query complexity
		\begin{align}
	\label{eq:int_pic_sim_cost}
	C_\mathrm{queries}[B,A,t,\epsilon]&=\mathcal{O}\left((t\alpha_B+1)C_\mathrm{queries}\left[B,A,\tau,\frac{\epsilon\tau}{t}\right]\right)
	\\\nonumber
	&=\mathcal{O}\left((t\alpha_B+1)
	\left(C_B+C_\mathrm{queries}\left[A,\tau,\frac{\epsilon\tau}{t}\right]\log{\lr{\frac{t\alpha_A}{\epsilon}}}\right)\frac{\log{(t/(\tau\epsilon))}}{\log\log{(t/(\tau\epsilon))}}
	\right)
	\\\nonumber
	&=\mathcal{O}\left((t\alpha_B+1)
	\left(C_B+C_\mathrm{queries}\left[A,\frac{1}{\alpha_B},\frac{\epsilon}{t\alpha_B}\right]\log{\lr{\frac{t\alpha_A}{\epsilon}}}\right)\log{\lr{\frac{t\alpha_B}{\epsilon}}}
	\right)
	\\\nonumber
	&=\mathcal{O}\left((t\alpha_B+1)
\left(C_B+C_\mathrm{queries}\left[A,\frac{1}{\alpha_B},\frac{\epsilon}{t\alpha_B}\right]\right)\log^2{\lr{\frac{t\alpha_A}{\epsilon}}}
\right),
\end{align}
where $\alpha_A\ge\alpha_B$ is used in the last line. Similarly, the gate complexity is
\begin{align}
	C_\mathrm{gates}[B,A,t,\epsilon]&=\mathcal{O}\left((t\alpha_B+1)C_\mathrm{gates}\left[B,A,\tau,\frac{\epsilon\tau}{t}\right]\right)
\\\nonumber
&=\mathcal{O}\left((t\alpha_B+1)
C_\mathrm{gates}\left[A,\tau,\frac{\epsilon\tau}{t}\right]\log{\lr{\frac{t\alpha_A}{\epsilon}}}\frac{\log{(t/(\tau\epsilon))}}{\log\log{(t/(\tau\epsilon))}}
\right)
\\\nonumber
&=\mathcal{O}\left((t\alpha_B+1)
C_\mathrm{gates}\left[A,\frac{1}{\alpha_B},\frac{\epsilon}{t\alpha_B}\right]\log^2{\lr{\frac{t\alpha_A}{\epsilon}}}
\right).
	\end{align}

\end{proof}

\section{Sparse matrix block-encoding by amplitude multiplication}
\label{sec:USA}
The cost of block-encoding a sparse Hamiltonian, given oracles $O_F$ and $O_H$ in~\cref{def:Sparse_Oracle} that compute its non-zero matrix positions and values to $b$ bits of precision, is given by~\cref{thm:sparse-Ham-block-encoding}. This is very similar to a procedure by~\cite{Low2017USA} that requires the coherent computation of inverse trigonometric functions and scales with $\mathcal{O}(\op{poly}(b))$. For completeness, we provide a proof outline of the original procedure with complexity given by~\cref{thm:sparse-Ham-block-encoding_trigonometric}. Subsequently, in~\cref{sec:arithmetic-free}, we prove~\cref{thm:sparse-Ham-block-encoding} which has linear scaling in $\mathcal{O}(b)$ by replacing computing trigonometric functions with just a quantum circuit for reversible addition.
\sparseoracles*
\begin{restatable}[Block encoding of sparse Hamiltonians by amplitude multiplication and quantum trigonometry~\cite{Low2017USA}]{lemma}{USAblockencodingTrig}
	\label{thm:sparse-Ham-block-encoding_trigonometric}
	Let	$H\in \mathbb{C}^{N\times N}$ be a $d$-sparse Hamiltonian satisfying all the following:
	\begin{itemize}
		\item There exist oracles $O_F$ and $O_H$ of~\cref{def:Sparse_Oracle} that compute the positions and values of non-zero matrix elements to $b$ bits of precision.
		\item An upper bound on the max-norm $\Lambda_{\mathrm{max}}\ge \|H\|_\mathrm{max} = \max_{ik}|H_{ik}|$ is known.
		\item An upper bound on the spectral-norm $\Lambda\ge \|H\| = \max_{v\neq 0}\frac{\|H\cdot v\|}{\|v\|}$ is known. 
		\item An upper bound on the induced one-norm $\Lambda_{1}\ge \|H\|_{1} = \max_{k}\sum_i|H_{ik}|$ is known.
	\end{itemize}
	Then there exists a Hamiltonian $\tilde H \in \mathbb{C}^{N\times N}$ that approximates $H$ with error $\|\tilde H- H\|=\mathcal{O}(\Lambda\delta)$, and can be block-encoded with normalizing constant $\alpha=\Theta(\Lambda_{1})$ using
	\begin{itemize}
		\item Queries $O_F$ and $O_H$: $\mathcal{O}\left(\sqrt{\frac{d\Lambda_{\mathrm{max}}}{\Lambda_{1}}}\log{\lr{\frac{1}{\delta}}}\right)$.
		\item Quantum gates: $\mathcal{O}\left(\sqrt{\frac{d\Lambda_{\mathrm{max}}}{\Lambda_{1}}}\log{\lr{\frac{1}{\delta}}}\left(\log{(N)}+b^k\right)\right)$, where $k=5/2$.
		\item Qubits: $\mathcal{O}(\log{(N)}+b)$.
	\end{itemize}
\end{restatable}
\begin{proof}[Proof of~\cref{thm:sparse-Ham-block-encoding_trigonometric}]
Let us construct two sets of mutually orthogonal quantum states $\{\ket{\chi_k}\}, \{\ket{\bar{\chi}_j}\}$ indexed by $j,k\in\{0,\cdots,N-1\}$, such that their overlap $\braket{\bar{\chi}_{j}}{\chi_{k}} = H_{jk} / \alpha$ is equal to the value of $H$ in the $j^{\text{th}}$ row and $k^{\text{th}}$ column, up to a normalizing constant $\alpha$. Then $H$ is block-encoded by the product of controlled-state preparation unitary $U = U_{\rm row}^\dag U_{\rm col}$, where
\begin{align}
U_{\rm col} \ket{k}_{s}\ket{0}_{a}=\ket{\chi_{k}}, \quad U_{\rm row} \ket{j}_{s}\ket{0}_{a}=\ket{\bar\chi_{j}}\quad 
\Rightarrow
\quad
(\bra{0}_a\otimes \ii_s)U(\ket{0}_a\otimes \ii_s) = \frac{H}{\alpha}.
\end{align}

One may verify that the following definition of these states achieves $\alpha=d\Lambda_\mathrm{max}$:
\begin{align}
\label{eq:sparse_states}
\ket{\chi_{k}}&= \ket{k}_s \frac{1}{\sqrt{d}} \sum_{p\in r_k}\ket{p}_{a_1} \left(\sqrt{\frac{H_{pk}}{\Lambda_{\max}}}\ket{0}_{a_2} + \sqrt{1-\frac{|H_{pk}|}{\Lambda_{\max}}}\ket{1}_{a_2} \right)\ket{0}_{a_3},
\\\nonumber
\bra{\bar{\chi}_{j}}&= \bra{j}_{a_1}\frac{1}{\sqrt{d}} \sum_{q\in r_j} \bra{q}_{s} \left(\sqrt{\frac{H_{jq}}{\Lambda_{\max}}}\bra{0}_{a_2} + \sqrt{1-\frac{|H_{jq}|}{\Lambda_{\max}}}\bra{2}_{a_2} \right)\bra{0}_{a_3},
\end{align}
where $r_k=\{f(i,k) : i\in[d]\}$ is the set of row indices of non-zero elements in the $k^{\text{th}}$ column, and the `$\mathrm{a}$' register is broken into subregisters `$\mathrm{a}_{1,2,3}$'. With this choice of $\ket{\chi_k}$, the unitary $U_{\rm col}$ is implemented by the following sequence:
\begin{align}\nonumber
\ket{k}_s\ket{0}_a
&\mapsto \ket{k}_s \frac{1}{\sqrt{d}}\sum_{\ell=1}^d\ket{\ell}_{a_1} \ket{0}_{a_2}\ket{0}_{a_3}
\underset{O_F}{\mapsto}\ket{k}_s \frac{1}{\sqrt{d}}\sum_{p\in r_k}\ket{p}_{a_1} \ket{0}_{a_2}\ket{0}_{a_3}
\underset{O_H}{\mapsto}\ket{k}_s \frac{1}{\sqrt{d}}\sum_{p\in r_k}\ket{p}_{a_1} \ket{0}_{a_2}\ket{H_{kp}}_{a_3}
\\\nonumber
&\underset{O_H}{\mapsto}\ket{k}_s \frac{1}{\sqrt{d}}\sum_{p\in r_k}\ket{p}_{a_1}\left(\sqrt{\frac{H_{pk}}{\Lambda_{\max}}}\ket{0}_{a_2} + \sqrt{1-\frac{|H_{pk}|}{\Lambda_{\max}}}\ket{1}_{a_2} \right) \ket{H_{kp}}_{a_3}
\\\nonumber
&\underset{O_H^{-1}}{\mapsto} \ket{k}_s \frac{1}{\sqrt{d}}\sum_{p\in r_k}\ket{p}_{a_1}\left(\sqrt{\frac{H_{pk}}{\Lambda_{\max}}}\ket{0}_{a_2} + \sqrt{1-\frac{|H_{pk}|}{\Lambda_{\max}}}\ket{1}_{a_2} \right)\ket{0}_{a_3}
\\
& = U_{\rm col}\ket{k}_s \ket{0}_a.
\label{eq:sparse_states_sequence}
\end{align}
The implementation of $U_{\rm row}$ is similar, with an additional step of swapping the `$\mathrm{s}$' and $`\mathrm{a}_1'$ registers using $\mathcal{O}(\log{(N)})$ quantum gates. In the first step, we use $\mathcal{O}(\log(d))=\mathcal{O}(\log{(N)})$ gates to prepare a uniform superposition over $d$ states. In the fourth step, we use $\mathcal{O}(\op{poly}(b))$ gates to compute $\theta=\sin^{-1}(\sqrt{|H_{kp}|/\Lambda_\mathrm{max}})$ and $\phi=\arg{[\sqrt{H_{kp}}]}$ using quantum arithmetic, and apply controlled single-qubit rotations to create the desired state in the `$\mathrm{a}_2$' register. Using a Taylor series and long multiplication, $\op{poly}(b)=\mathcal{O}(b^{5/2})$~\cite{Berry2015Hamiltonian}. Thus $U_{\rm row}^\dag U_{\rm col}$ makes $\mathcal{O}(1)$ queries to $O_F$ and $O_H$, and requires $\mathcal{O}(\log(N)+\op{poly}(b))$ quantum gates. 

We now modify this implementation to block-encode $H$ with a smaller normalization constant $\alpha=\Theta(\Lambda_1)$.  We focus on $\ket{\chi_{k}}$ as the modification for $\ket{\bar\chi_{j}}$ is very similar. By collecting coefficients of $\ket{0}_{a_2}$, ~\cref{eq:sparse_states} is equivalent to
\begin{align}
\label{eq:sparse_states_collected}
\ket{\chi_{k}}&= \sqrt{\frac{\sigma_k}{d\Lambda_{\max}}}\ket{k}_s\left(\sum_{p\in r_k}\sqrt{\frac{H_{pk}}{\sigma_k}} \ket{p}_{a_1} \right)\ket{0}_{a_2}\ket{0}_{a_3}+\cdots \ket{1}_{a_2},
\end{align}
where $\sigma_k = \sum_{k}|H_{pk}|= \sum_{k}|H_{kp}|$. Thus the terms in round brackets are normalized quantum states. As the induced one-norm satisfies $\|H\|_1=\max_k \sigma_k\le d\|H\|_\mathrm{max}\le d\Lambda_\mathrm{max}$, the amplitude $\sqrt{\frac{\sigma_k}{d\Lambda_{\max}}}$ of $\ket{0}_{a_2}$ is in general less than $1$. Thus any procedure that multiplies these amplitudes by some constant $C \le \sqrt{\frac{d\Lambda_{\max}}{\Lambda_1}}$ for all $j,k$ would block encode $H$ with normalization constant $\alpha=\frac{C^2}{d\Lambda_\mathrm{max}}$.

This is accomplished by amplitude multiplication~\cite{Low2017USA}, which is a high-precision variant of amplitude amplification. Let us briefly compare the two approaches. Amplitude amplification uses $\mathcal{O}(C)$ queries to $U_{\rm col}$ and $\mathcal{O}(C\log{(N)})$ quantum gates for implementing reflections about the $\ket{0}_a$ state to transform the amplitudes $a_k=\sqrt{\frac{\sigma_k}{d\Lambda_{\max}}}$ nonlinearly like $a_k\rightarrow\sin{\left(C\sin^{-1}\left(a_k\right)\right)}$ for odd integers $C$. When $|C a_k| \ll 1$, this is approximately linear in the input $a_k$. In contrast, amplitude multiplication requires any upper bound $A \ge |a_k|$ and uses $\mathcal{O}(C\log{(1/\delta)})$ queries to  $U_{\rm col}$ and $\mathcal{O}(C\log{(1/\delta)}\log{(N)})$ quantum gates to multiply any unknown $a_k$ by any choice of real number $C=\mathcal{O}(1/A)$ with some bounded real multiplicative error $|\delta_k|\le \delta$ that is a function of $a_k$. In other words,
\begin{align}
\sqrt{\frac{\sigma_k}{d\Lambda_{\max}}} \rightarrow_\text{Amplitude multiplication} C \sqrt{\frac{\sigma_k}{d\Lambda_{\max}}}(1+\delta_k), \quad |\delta_k|\le \delta,
\end{align}
which can be made arbitrarily linear with respect to the initial amplitude, at logarithmic cost.

By choosing $C=\Theta(\sqrt{\frac{d\Lambda_{\max}}{\Lambda_1}})$, we make $\mathcal{O}(\sqrt{\frac{d\Lambda_{\max}}{\Lambda_1}}\log{(1/\delta)})$ queries to $U_\mathrm{col}$ and $U_\mathrm{row}$, and a multiplicative factor $\mathcal{O}(\log{(N)})$ more quantum gates, to block-encode a Hamiltonian $\tilde{H}$ with matrix elements $\frac{\tilde{H}_{j,l}}{\alpha}=\frac{H_{jk}}{\alpha}(1+\delta_j)(1+\delta_k)$. Let $\hat{\delta}$ be a diagonal matrix with elements $\delta_j$. Then the block-encoded Hamiltonian $\tilde{H}/\alpha$ differs from the ideal Hamiltonian $H/\alpha$ by an error
\begin{align}
\|\tilde{H}-H\|
=\|\hat\delta H+H\hat\delta+\hat\delta H \hat\delta\|
 \le \|H\|(2\delta+\delta^2)=\mathcal{O}(\Lambda\delta).
\end{align}
As $U_\mathrm{col}$ and $U_\mathrm{row}$ each make $\mathcal{O}(1)$ queries to $O_F$ and $O_H$ and each require $\mathcal{O}(\log(N)+\op{poly}(b))$ quantum gates, multiplying these with $\mathcal{O}(\sqrt{\frac{d\Lambda_{\max}}{\Lambda_1}}\log{(1/\delta)})$ gives the stated query and gate complexities.
\end{proof}

\subsection{Quantum addition instead of inverse trigonometric functions}
\label{sec:arithmetic-free}
The proof~\cref{thm:sparse-Ham-block-encoding} requires us to specify the format in which matrix entries $\ket{H_{ik}}$ are returned by the oracle $O_H$ of~\cref{def:Sparse_Oracle}. We use the following definition.
\begin{restatable}[Qubit encoding of complex fixed-point numbers]{definition}{fixedpoint}
	\label{def:fixed_point}
	Let $z=r e^{i 2\pi\phi}$ be a complex number, where $\phi\in[0,1)$ is represented by $p$ fractional bits $\vec\phi \in\{0,1\}^{p}$ and	$r\in [0,2^{m}]$ is represented by $m+1$ integer bits and $n$ fractional bits $\vec r \in\{0,1\}^{m+n+1}$ as follows.
	\begin{align}
	\phi=\sum^{p}_{j=1}\phi_{j}2^{-j},
	\quad
	r=2^{m}\sum^{m+n}_{j=0}r_{j}2^{-j}.
	\end{align}
	Then $\ket{z_{p,m,n}}$ is the $p+m+n+1$ qubit state
	\begin{align}
	\ket{z_{p,m,n}}=\ket{\phi_1}\cdots\ket{\phi_p}\ket{r_0}\cdots\ket{r_{m+n}}=\ket{\phi}_p\ket{r2^n}_{mn}.
	\end{align}
\end{restatable}
Thus the oracle $O_H$ computes $H_{ik}$ to $b$ bits of precision in the format of~\cref{def:fixed_point} like
\begin{align}
O_H\ket{i}\ket{k}\ket{z}=\ket{i}\ket{k}\ket{z\oplus H_{ik}}=\ket{i}\ket{k}\ket{z\oplus (H_{ik})_{p,m,n}},
\end{align}
where the bits of precision $b=p+m+n+1$. This is a valid starting assumption -- any binary representation of $\ket{H_{ik}}$ is polynomial-time reducible to $\ket{(H_{ik})_{p,m,n}}$. 

Given any complex number $H_{ik}$ in this format and an upper bound $\Lambda_\mathrm{max}\ge \|H\|_\mathrm{max}$, the goal is to coherently apply $\sqrt{\frac{H_{ik}}{\Lambda_\mathrm{max}}}$ as an amplitude. Unlike the proof of~\cref{thm:sparse-Ham-block-encoding_trigonometric} which required the costly coherent computation of inverse trigonometric functions, this may be performed using only simple quantum arithmetic as follows from~\cite{Sanders2018}.
\begin{restatable}[Complex fixed-point numbers to amplitudes]{lemma}{fixedpoint}
	\label{def:fixed_point_to_amplitude}
	Let the complex number $z=r e^{i 2\pi\phi}$ where $\phi\in[0,1)$ and $r\in [0,2^{m}]$ be represented by $\ket{z_{p,m,n}}=\ket{\phi_1}\cdots\ket{\phi_{p}}\ket{r_0}\cdots\ket{r_{m+n}}$. Let 
	$\Lambda_{\mathrm{max}}\ge 2^m$. Then the transformation
	\begin{align}
	\ket{z_{p,m,n}}\ket{0}_a\ket{0}_b\ket{0}_c\mapsto\ket{z_{p,m,n}}\left(\sqrt{\frac{z}{\Lambda_{\mathrm{max}}}}\ket{u_{|z|}}_a\ket{00}_{bc}+\cdots\ket{\Phi}_{abc}\right), \quad|(I_a\otimes\bra{0}_{bc})\ket{\Phi}_{abc}|=0,
	\end{align}	
	where $\ket{u_{|z|}}_a=\sum^{r2^n}_{j=1}\frac{\ket{j}}{\sqrt{r2^n}}$, $a$ is a $m+n$ qubit register, and $b,c$ are single-qubit registers, costs $\mathcal{O}(p+n+m)$ arbitrary two qubit gates.
\end{restatable}
\begin{proof}
	Consider the following sequence
	\begin{align}
	\ket{z_{p,m,n}}\ket{0}_a\ket{0}_b&=\ket{\phi_1}\cdots\ket{\phi_p}\ket{r_0}\cdots\ket{r_{m+n}}\ket{0}_a\ket{0}_b\ket{0}_c=\ket{\phi}_{p}\ket{r}_{mn}\ket{0}_a\ket{0}_b\ket{0}_c
	\displaybreak[0]\\\nonumber
	\underset{\op{UNIFORM}}{\;}&{\mapsto} \ket{\phi}_{p}\ket{r2^n}_{mn} \left(\sum^{2^{m+n}}_{j=1}\frac{\ket{j}_a}{\sqrt{2^{m+n}}}\right) \ket{0}_b\ket{0}_c
	\displaybreak[0]\\\nonumber
	\underset{\op{COMPARE}}{\;}&{\mapsto} \ket{\phi}_{p}\ket{r2^n}_{mn} \left(\sum^{2^{m+n}}_{j=1}\frac{\ket{j}_a}{\sqrt{2^{m+n}}}\right) \ket{r2^n < j}_b\ket{0}_c
	\\\nonumber
	&= \ket{\phi}_{p}\ket{r2^n}_{mn} \left(\left(\sum^{r2^n}_{j=1}\frac{\ket{j}_a}{\sqrt{2^{m+n}}}\right) \ket{0}_b+\left(\sum^{2^{m+n}}_{j=r2^n+1}\frac{\ket{j}_a}{\sqrt{2^{m+n}}}\right) \ket{1}_b\right)\ket{0}_c
	\\\nonumber
	&= \ket{\phi}_{p}\ket{r2^n}_{mn} \left(\sqrt{\frac{r}{2^m}}\ket{u_{r}}_a \ket{0}_b+\cdots\ket{1}_b\right)\ket{0}_c
	\displaybreak[0]\\\nonumber
	\underset{\op{PHASE}}{\;}&\mapsto \ket{\phi}_{p}\ket{r2^m}_{mn}  \left(\sqrt{\frac{re^{i2\pi\phi}}{2^m}}\ket{u_{r}}_a\ket{0}_b+\cdots\ket{1}_b\right)\ket{0}_c
	\displaybreak[0]\\\nonumber
	\underset{e^{-iY\cos^{-1}(2^m/\Lambda_{\max})}}{\;}&\mapsto \ket{\phi}_{p}\ket{r2^m}_{mn}  \left(\sqrt{\frac{re^{i2\pi\phi}}{2^m}}\ket{u_{r}}_a\ket{0}_b+\cdots\ket{1}_b\right)\left(\sqrt{\frac{2^m}{\Lambda_{\mathrm{max}}}}\ket{0}_c+\sqrt{1-\frac{2^m}{\Lambda_{\mathrm{max}}}}\ket{1}_c\right)
	\\\nonumber
	&= \ket{z_{p,m,n}} \left(\sqrt{\frac{re^{i2\pi\phi}}{\Lambda_{\mathrm{max}}}}\ket{u_{r}}_a\ket{0}_{bc}+\cdots\ket{\Phi}_{abc}\right).
	\end{align}
	The second line creates a uniform superposition over $2^{n+m}$ states using $n+m$ Hadamard gates. The third line compares two integer registers $\ket{r2^n}_{mn}$ and $\ket{j}_{a}$. If $r2^n< j$, a $X$ gate is applied to create the state $\ket{r2^m< j}_b=\ket{1}_b$. Otherwise the state $\ket{0}_b$ is unchanged. This uses a subtractor circuit (the conjugate of modular addition) and costs $\mathcal{O}(n+m)$ arbitrary two-qubit gates~\cite{Cuccaro2004Adder}. The fourth line collects coefficients of $\ket{0}_b$, which is expressed in the fifth line as a normalized quantum state
	\begin{align}
	\label{eq:uniform_superposition}
	\ket{u_r}_a=\sum^{r2^n}_{j=1}\frac{\ket{j}_a}{\sqrt{r2^n}}.
	\end{align}
	The sixth line applies $p$ single-qubit phase rotations on register $b$, where the $j^\text{th}$ rotation $e^{i Z \pi \phi_j 2^{-j}}$ is controlled by the state $\ket{\phi_j}$. 	The seventh line applies a single-qubit $Y$ Pauli rotation. The last line collects coefficients of the $\ket{00}_{bc}$ state. All other components are in $\ket{\Phi}_{abc}$, which has no support on $\ket{00}_{bc}$.
	\end{proof}
	We now prove~\cref{thm:sparse-Ham-block-encoding}, restated for convenience.
	\USAblockencoding*
	\begin{proof}
		The proof is identical to~\cref{thm:sparse-Ham-block-encoding_trigonometric}, except for a modification to~\cref{eq:sparse_states} and~\cref{eq:sparse_states_sequence}. As before, we construct two sets of mutually orthogonal quantum states such that their overlap $\braket{\bar{\chi}_{j}}{\chi_{k}} = H_{jk} / \alpha$. Then $H$ is block-encoded by the product of controlled-state preparation unitary $U = U_{\rm row}^\dag U_{\rm col}$, where
		\begin{align}
		U_{\rm col} \ket{k}_{s}\ket{0}_{a}=\ket{\chi_{k}}, \quad U_{\rm row} \ket{j}_{s}\ket{0}_{a}=\ket{\bar\chi_{j}}\quad 
		\Rightarrow
		\quad
		(\bra{0}_a\otimes \ii_s)U(\ket{0}_a\otimes \ii_s) = \frac{H}{\alpha}.
		\end{align}
		We now slightly modify the definition of $\ket{\chi_{k}}, \ket{\bar\chi_{k}}$ compared to~\cref{eq:sparse_states}:
		\begin{align}
		\ket{\chi_{k}}&= \ket{k}_s \sum_{p\in r_k}\frac{\ket{p}_{a_1}}{\sqrt{d}} \left(\sqrt{\frac{H_{pk}}{\Lambda_{\max}}}\ket{u_{|H_{pk}|}}_{a_4}\ket{0}_{a_2} +\cdots \ket{\Phi}_{a_4}\ket{1}_{a_2} \right)\ket{0}_{a_3},
		\\\nonumber
		\bra{\bar{\chi}_{j}}&= \bra{j}_{a_1} \sum_{q\in r_j} \frac{\bra{q}_{s}}{\sqrt{d}} \left(\sqrt{\frac{H_{jq}}{\Lambda_{\max}}}\bra{u_{|H_{jq}|}}_{a_4}\bra{0}_{a_2} +\cdots \bra{\bar\Phi}_{a_4}\bra{2}_{a_2}\right)\bra{0}_{a_3},
		\end{align}
		where we have introduced the additional state $\ket{u_{|H_{pk}|}}_{a_4}$ of~\cref{eq:uniform_superposition}, and $\ket{\Phi}$ and $\ket{\bar\Phi}$ are arbitrary quantum states of no interest. As $\ket{u_{|H_{jk}|}}_{a_4}$ only depends on the absolute value $|H_{jk}|$, all case where $k=q$ and $p=j$ also have $\braket{u_{|H_{pk}|}}{u_{|H_{jq}|}}_{a_4}=1$. Thus the overlap $\frac{H_{jk}}{d\Lambda_\mathrm{max}}$ is unaffected.
		
		Now instead of constructing $U_{\rm col}$ by computing trigonometric functions, as in~\cref{eq:sparse_states}, we use~\cref{def:fixed_point_to_amplitude} to convert a binary encoding of $H_{jk}$ to an amplitude. Consider the modified sequence:
		\begin{align}\nonumber
		\ket{k}_s\ket{0}_a
		&
		\underset{O_F}{\mapsto}\ket{k}_s \frac{1}{\sqrt{d}}\sum_{p\in r_k}\ket{p}_{a_1} \ket{0}_{a_4}\ket{0}_{a_2}\ket{0}_{a_3}
		\underset{O_H}{\mapsto}\ket{k}_s \frac{1}{\sqrt{d}}\sum_{p\in r_k}\ket{p}_{a_1} \ket{0}_{a_4}\ket{0}_{a_2}\ket{H_{kp}}_{a_3}
		\\\nonumber
		{{}_{\cref{def:fixed_point_to_amplitude}}}&{\;\mapsto}\ket{k}_s \frac{1}{\sqrt{d}}\sum_{p\in r_k}\ket{p}_{a_1}\left(\sqrt{\frac{H_{pk}}{\Lambda_{\max}}}\ket{u_{|H_{pk}|}}_{a_4}\ket{0}_{a_2} +\cdots\ket{\Phi}_{a_4}\ket{1}_{a_2} \right) \ket{H_{kp}}_{a_3}
		\\\nonumber
		&\underset{O_H^{-1}}{\mapsto} \ket{k}_s \frac{1}{\sqrt{d}}\sum_{p\in r_k}\ket{p}_{a_1}\left(\sqrt{\frac{H_{pk}}{\Lambda_{\max}}}\ket{u_{|H_{pk}|}}_{a_4}\ket{0}_{a_2}  +\cdots\ket{\Phi}_{a_4}\ket{1}_{a_2} \right)\ket{0}_{a_3}
		\\
		& = U_{\rm col}\ket{k}_s \ket{0}_a,
		\end{align}
		and similarly for $U_\mathrm{row}$. From~\cref{def:fixed_point_to_amplitude}, the second line now uses $\mathcal{O}(b)$ arbitrary two-qubit gates instead of the original $\mathcal{O}(b^{5/2})$. The rest of the proof proceeds identically.
	\end{proof}

\bibliographystyle{alphaUrlePrint}
\bibliography{bib}

\newcommand{\etalchar}[1]{$^{#1}$}
\begin{thebibliography}{CDKM04}

\bibitem[ATS03]{Aharonov2003Adiabatic}
Dorit Aharonov and Amnon Ta-Shma.
\newblock \href{http://dx.doi.org/10.1145/780542.780546}{Adiabatic quantum
  state generation and statistical zero knowledge}.
\newblock In {\em Proceedings of the Thirty-fifth Annual ACM Symposium on
  Theory of Computing}, STOC '03, pages 20--29, New York, NY, USA, 2003. ACM.

\bibitem[BACS07]{Berry2007Efficient}
Dominic~W. Berry, Graeme Ahokas, Richard Cleve, and Barry~C. Sanders.
\newblock \href{http://dx.doi.org/10.1007/s00220-006-0150-x}{Efficient quantum
  algorithms for simulating sparse {H}amiltonians}.
\newblock {\em Communications in Mathematical Physics}, 270(2):359--371, 2007.

\bibitem[BC12]{Berry2012}
Dominic~W. Berry and Andrew~M. Childs.
\newblock \href{http://dl.acm.org/citation.cfm?id=2231036.2231040}{Black-box
  {H}amiltonian simulation and unitary implementation}.
\newblock {\em Quantum Info. Comput.}, 12(1-2):29--62, January 2012.

\bibitem[BCC{\etalchar{+}}14]{Berry2014}
Dominic~W. Berry, Andrew~M. Childs, Richard Cleve, Robin Kothari, and
  Rolando~D. Somma.
\newblock \href{http://dx.doi.org/10.1145/2591796.2591854}{Exponential
  improvement in precision for simulating sparse {H}amiltonians}.
\newblock In {\em Proceedings of the 46th Annual ACM Symposium on Theory of
  Computing}, STOC '14, pages 283--292, New York, NY, USA, 2014. ACM.

\bibitem[BCC{\etalchar{+}}15]{Berry2015Truncated}
Dominic~W. Berry, Andrew~M. Childs, Richard Cleve, Robin Kothari, and
  Rolando~D. Somma.
\newblock \href{http://dx.doi.org/10.1103/PhysRevLett.114.090502}{Simulating
  {H}amiltonian dynamics with a truncated {T}aylor series}.
\newblock {\em Phys. Rev. Lett.}, 114:090502, Mar 2015.

\bibitem[BCK15]{Berry2015Hamiltonian}
Dominic~W. Berry, Andrew~M. Childs, and Robin Kothari.
\newblock \href{http://dx.doi.org/10.1109/FOCS.2015.54}{Hamiltonian simulation
  with nearly optimal dependence on all parameters}.
\newblock In {\em Foundations of Computer Science (FOCS), 2015 IEEE 56th Annual
  Symposium on}, pages 792--809, Oct 2015.

\bibitem[BN16]{Berry2016corrected}
Dominic~W. Berry and Leonardo Novo.
\newblock \href{http://dl.acm.org/citation.cfm?id=3179439.3179442}{Corrected
  quantum walk for optimal {H}amiltonian simulation}.
\newblock {\em Quantum Info. Comput.}, 16(15-16):1295--1317, November 2016.

\bibitem[CCD{\etalchar{+}}03]{Childs2003Exponential}
Andrew~M. Childs, Richard Cleve, Enrico Deotto, Edward Farhi, Sam Gutmann, and
  Daniel~A. Spielman.
\newblock \href{http://dx.doi.org/10.1145/780542.780552}{Exponential
  algorithmic speedup by a quantum walk}.
\newblock In {\em Proceedings of the Thirty-fifth Annual ACM Symposium on
  Theory of Computing}, STOC '03, pages 59--68, New York, NY, USA, 2003. ACM.

\bibitem[CDKM04]{Cuccaro2004Adder}
Steven~A. Cuccaro, Thomas~G. Draper, Samuel~A. Kutin, and David~Petrie Moulton.
\newblock \href{https://arxiv.org/abs/quant-ph/0410184}{A new quantum
  ripple-carry addition circuit}.
\newblock {\em arXiv preprint quant-ph/0410184}, 2004.

\bibitem[CGJ18]{Chakraborty2018BlockEncoding}
Shantanav Chakraborty, Andr{\'a}s Gily{\'e}n, and Stacey Jeffery.
\newblock \href{https://arxiv.org/abs/1804.01973}{The power of block-encoded
  matrix powers: improved regression techniques via faster {H}amiltonian
  simulation}.
\newblock {\em arXiv preprint arXiv:1804.01973}, 2018.

\bibitem[Chi10]{Childs2010}
Andrew~M. Childs.
\newblock \href{http://dx.doi.org/10.1007/s00220-009-0930-1}{On the
  relationship between continuous- and discrete-time quantum walk}.
\newblock {\em Commun. Math. Phys.}, 294(2):581--603, 2010.

\bibitem[CK10]{Childs2010Limitation}
Andrew~M. Childs and Robin Kothari.
\newblock \href{http://dl.acm.org/citation.cfm?id=2011373.2011380}{Limitations
  on the simulation of non-sparse {H}amiltonians}.
\newblock {\em Quantum Info. Comput.}, 10(7):669--684, July 2010.

\bibitem[CK11]{Childs2011}
Andrew~M. Childs and Robin Kothari.
\newblock \href{http://dx.doi.org/10.1007/978-3-642-18073-6_8}{{\em Theory of
  Quantum Computation, Communication, and Cryptography}}, pages 94--103.
\newblock Springer Berlin Heidelberg, 2011.

\bibitem[CKS17]{Childs2015LinearSystems}
Andrew~M. Childs, Robin Kothari, and Rolando~D. Somma.
\newblock \href{http://dx.doi.org/10.1137/16M1087072}{Quantum algorithm for
  systems of linear equations with exponentially improved dependence on
  precision}.
\newblock {\em SIAM Journal on Computing}, 46(6):1920--1950, 2017.

\bibitem[CMN{\etalchar{+}}17]{Childs2017Speedup}
Andrew~M. Childs, Dmitri Maslov, Yunseong Nam, Neil~J. Ross, and Yuan Su.
\newblock \href{https://arxiv.org/abs/1711.10980}{Toward the first quantum
  simulation with quantum speedup}.
\newblock {\em arXiv preprint arXiv:1711.10980}, 2017.

\bibitem[CW12]{Childs2012}
Andrew~M. Childs and Nathan Wiebe.
\newblock \href{http://dl.acm.org/citation.cfm?id=2481569.2481570}{Hamiltonian
  simulation using linear combinations of unitary operations}.
\newblock {\em Quantum Info. Comput.}, 12(11-12):901--924, November 2012.

\bibitem[Fey82]{Feynman1982}
Richard~P. Feynman.
\newblock \href{http://dx.doi.org/10.1007/BF02650179}{Simulating physics with
  computers}.
\newblock {\em International Journal of Theoretical Physics}, 21(6):467--488,
  1982.

\bibitem[GLM08]{Giovannetti2009qRAM}
Vittorio Giovannetti, Seth Lloyd, and Lorenzo Maccone.
\newblock \href{http://dx.doi.org/10.1103/PhysRevLett.100.160501}{Quantum
  random access memory}.
\newblock {\em Phys. Rev. Lett.}, 100:160501, Apr 2008.

\bibitem[GSLW18]{Gilyen2018quantum}
Andr{\'a}s Gily{\'e}n, Yuan Su, Guang~Hao Low, and Nathan Wiebe.
\newblock \href{https://arxiv.org/abs/1806.01838}{Quantum singular value
  transformation and beyond: exponential improvements for quantum matrix
  arithmetics}.
\newblock {\em arXiv preprint arXiv:1806.01838}, 2018.

\bibitem[Haa18]{Haah2018product}
Jeongwan Haah.
\newblock \href{https://arxiv.org/abs/1806.10236}{Product decomposition of
  periodic functions in quantum signal processing}.
\newblock {\em arXiv preprint arXiv:1806.10236}, 2018.

\bibitem[HHKL18]{Haah2018quantum}
Jeongwan Haah, Matthew~B Hastings, Robin Kothari, and Guang~Hao Low.
\newblock \href{https://arxiv.org/abs/1801.03922}{Quantum algorithm for
  simulating real time evolution of lattice {H}amiltonians}.
\newblock {\em arXiv preprint arXiv:1801.03922}, 2018.

\bibitem[HHL09]{Harrow2009}
Aram~W. Harrow, Avinatan Hassidim, and Seth Lloyd.
\newblock \href{http://dx.doi.org/10.1103/PhysRevLett.103.150502}{Quantum
  algorithm for linear systems of equations}.
\newblock {\em Phys. Rev. Lett.}, 103:150502, Oct 2009.

\bibitem[HK18]{Harrow2018}
Aram~W. Harrow and Robin Kothari.
\newblock Private communication.
\newblock 2018.

\bibitem[JW09]{Jordan2009}
Stephen~P. Jordan and Pawel Wocjan.
\newblock \href{http://dx.doi.org/10.1103/PhysRevA.80.062301}{Efficient quantum
  circuits for arbitrary sparse unitaries}.
\newblock {\em Phys. Rev. A}, 80:062301, Dec 2009.

\bibitem[LC16]{Low2016hamiltonian}
Guang~Hao Low and Isaac~L Chuang.
\newblock \href{https://arxiv.org/abs/1610.06546}{Hamiltonian simulation by
  qubitization}.
\newblock {\em arXiv preprint arXiv:1610.06546}, 2016.

\bibitem[LC17a]{Low2017USA}
Guang~Hao Low and Isaac~L Chuang.
\newblock \href{https://arxiv.org/abs/1707.05391}{Hamiltonian simulation by
  uniform spectral amplification}.
\newblock {\em arXiv preprint arXiv:1707.05391}, 2017.

\bibitem[LC17b]{Low2016HamSim}
Guang~Hao Low and Isaac~L. Chuang.
\newblock \href{http://dx.doi.org/10.1103/PhysRevLett.118.010501}{Optimal
  {H}amiltonian simulation by quantum signal processing}.
\newblock {\em Phys. Rev. Lett.}, 118:010501, Jan 2017.

\bibitem[Llo96]{Lloyd1996universal}
Seth Lloyd.
\newblock
  \href{http://search.proquest.com/docview/213562780?accountid=12492}{Universal
  quantum simulators}.
\newblock {\em Science}, 273(5278):1073, Aug 23 1996.

\bibitem[LW18]{Low2018IntPicSim}
Guang~Hao Low and Nathan Wiebe.
\newblock \href{https://arxiv.org/abs/1805.00675}{Hamiltonian simulation in the
  interaction picture}.
\newblock {\em arXiv preprint arXiv:1805.00675}, 2018.

\bibitem[LYC16]{Low2016}
Guang~Hao Low, Theodore~J. Yoder, and Isaac~L. Chuang.
\newblock \href{http://dx.doi.org/10.1103/PhysRevX.6.041067}{Methodology of
  resonant equiangular composite quantum gates}.
\newblock {\em Phys. Rev. X}, 6:041067, Dec 2016.

\bibitem[Osb12]{Osborne2012}
Tobias~J Osborne.
\newblock \href{http://stacks.iop.org/0034-4885/75/i=2/a=022001}{Hamiltonian
  complexity}.
\newblock {\em Reports on Progress in Physics}, 75(2):022001, 2012.

\bibitem[Rei09]{Reichardt2009span}
Ben~W. Reichardt.
\newblock \href{http://dx.doi.org/10.1109/FOCS.2009.55}{Span programs and
  quantum query complexity: The general adversary bound is nearly tight for
  every boolean function}.
\newblock In {\em Proceedings of the 2009 50th Annual IEEE Symposium on
  Foundations of Computer Science}, FOCS '09, pages 544--551, Washington, DC,
  USA, 2009. IEEE Computer Society.

\bibitem[SLSB18]{Sanders2018}
Yuval~R. Sanders, Guang~Hao Low, Artur Scherer, and Dominic~W. Berry.
\newblock \href{https://arxiv.org/abs/1807.03206}{Black-box quantum state
  preparation without arithmetic}.
\newblock {\em arXiv preprint arXiv:1807.03206}, 2018.

\bibitem[WW18]{Wang2018NonSparse}
Chunhao Wang and Leonard Wossnig.
\newblock \href{https://arxiv.org/abs/1803.08273}{A quantum algorithm for
  simulating non-sparse {H}amiltonians}.
\newblock {\em arXiv preprint arXiv:1803.08273}, 2018.

\end{thebibliography}

\end{document}